%% file: corr_main.tex
\documentclass{llncs}
\input{macros}

\def\sharedaffiliation{%
\end{tabular}
\begin{tabular}{c}}

\begin{document}

%\title{ZaliQL: A SQL-Based Framework for Drawing Causal Inference from Big Data}
\ignore{
\numberofauthors{2}
\author{
  \alignauthor Babak Salimi\\
  \alignauthor Dan Suciu
  \sharedaffiliation
  \affaddr{Department of Computer Science \& Engineering} \\
  \affaddr{University of Washington} \\
  \affaddr{\{bsalimi, suciu\}@cs.washington.edu}
}}
% There's nothing stopping you putting the seventh, eighth, etc.
% author on the opening page (as the 'third row') but we ask,
% for aesthetic reasons that you place these 'additional authors'
% in the \additional authors block, viz.fff
% Just remember to make sure that the TOTAL number of authors
% is the number that will appear on the first page PLUS the
% number that will appear in the \additionalauthors section.

     % if too long for running head

\title{\vspace*{-1cm} {\bf ZaliQL: A SQL-Based Framework for Drawing Causal Inference from Big Data}}
\author{{\bf Babak Salimi} \ \ and \ \ {\bf Dan Suciu}\\
Department of Computer Science \& Engineering\\  Seattle,  USA\\
\hspace*{1.5cm}\{bsalimi, suciu\}@cs.washington.edu
}
\institute{University of Washington}
\maketitle

\begin{abstract}
Causal inference from observational data is a subject of active research and development in statistics and computer science. Many toolkits have been developed for this purpose that
 depends on  statistical software. However, these toolkits do not scale to large datasets. In this paper we  describe a suite of techniques for expressing causal inference tasks from observational data
  in SQL.  This suite supports the state-of-the-art methods for causal inference and run at scale
  within a database engine. In addition, we introduce
  several optimization techniques that significantly speedup causal
  inference, both in the online and offline setting.  We evaluate the quality and performance of our techniques by experiments of real datasets.

\end{abstract}

\input{intro}
\input{background}

\input{btech}
\input{optimization}

\input{experimental}

\input{related}
\vspace{-.45cm}
\bibliographystyle{abbrv}
\bibliography{ref}
\end{document}

%% file: macros.tex
\usepackage{pdfpages}
\usepackage{graphicx}%\usepackage[draft]{hyperref}
\ifpdf
  \DeclareGraphicsExtensions{.pdf,.png,.jpg}
\else
  \DeclareGraphicsExtensions{.eps}
\fi
\usepackage{epstopdf}
\usepackage{url}

\usepackage{hyperref,caption,subcaption}
\usepackage{alltt}
\usepackage{algpseudocode}
\usepackage{amsmath,amssymb}
\usepackage{algorithm}
\usepackage{algpseudocode}
\def\sharedaffiliation{%
\end{tabular}
\begin{tabular}{c}}

\newcommand{\bigCI}{\mathrel{\text{\scalebox{1.07}{$\perp\mkern-10mu\perp$}}}}

\newcommand{\mc}[1]{\mathcal{#1}}
\newcommand{\ignore}[1]{}

\newcommand{\dan}[1]{{{\color{magenta} Dan: [{#1}]}}}
\newcommand{\babak}[1]{{\texttt{\color{blue} Babak: [{#1}]}}}

\newcommand{\ate}{{\tau_{ATE}}}

\newcommand{\rel}{{{R}}}

\newcommand{\rele}{{{R}^e}}
\newcommand{\relei}{{{R}^e_{T_i}}}
\newcommand{\relt}{{{R}^e_{\trep}}}

\newcommand{\crele}{{{R}^c}}

\newcommand{\cem}{{\textsc{Cem}}}
\newcommand{\mcem}{{\textsc{mCem}_{T_i}}}
\newcommand{\nnmwr}{{\textsc{Nnmwr}}}
\newcommand{\nnmnr}{{\textsc{Nnmnr}}}
\newcommand{\sbc}{{\textsc{SubC}}}
\newcommand{\prp}{{{P}_{\trep}}}
\newcommand{\prpi}{{{P}_{T_i}}}
\newcommand{\cv}{{{X}}}
\newcommand{\ccv}{{{\mc{X}}}}
\newcommand{\ccvi}{{{{cx}}}}

\newcommand{\E}{\mathbb{E}}
\newcommand{\GSQL}{{{\textrm{ZaliQL}}}}
\newcommand{\delay}{\textsc{FlightDelay}}
\newcommand{\tre}{{{\mc{T}}}}
\newcommand{\trep}{S}

\newcommand{\ccvin}{{{\mc{X}'}}}
\newcommand{\ccvu}{{{\mc{X}}}}

\newcommand{\NLOGSPACE}{{\scriptsize $  \mathrm{NLOGSPACE}$}}
\newcommand{\PTIME}{{\scriptsize $  \mathrm{PTIME}$ }}

\newcommand{\proj}[1]{{\Pi}}
\newcommand{\sel}[1]{{\sigma}}

\newcommand{\cut}[1]{}
\newcommand{\eat}[1]{}

\newcommand{\commentresolved}[1]{}

\newcommand{\ie}{{\em i.e.}} %\xspace}
\newcommand{\eg}{{\em e.g.}} %\xspace}
\newcommand{\set}[1]{\{#1\}}                    % Set (as in \set{1,2,3}).
\newcommand{\setof}[2]{\{{#1}\mid{#2}\}}        % Set (as in \setof{x}{x>0}).

%% file: intro.tex
\section{Introduction}
%\dan{7 pages}
\label{sec:introduction}

%\dan{this part is mostly done, only minor edits needed}

% Big data is used today in a wide range of domains, such as natural
% sciences (physics CITE, astronomy CITE, genomics CITE, biology CITE),
% social science and in particular measurements of human activity
% (e.g. recommendation CITE, personalization CITE),
% education~\cite{clauset-2015}, marketing and economics CITE, and MORE
% HERE.  The most successful type of application of Big data is {\em
%   predictive}: the data available, the \emph{seen} data, is used to
% infer a model that can predict features and trends of data that is not
% available, the \emph{unseen} data.  This potential for predictive
% analysis has generated the huge interest we see today in Big data, and
% has, in some sense, democratized data, by stimulating a huge number of
% ``data enthusiasts'' to collect, explore, analyze, integrate, and
% visualize data.  The term Big data is often used today to refer not
% just to the traditional features like volume, variety, velocity, but
% also to its wide availability to a broad range of users.

% lg1 - original commentout below at end; take care as this can lead to confusion

Much of the success of Big data today comes from \emph{predictive or descriptive
  analytics}:  statistical models or data mining algorithms applied to data
to predict new or future observations, e.g., we observe
how users click on ads, then build a model and predict how future
users will click.  Predictive analysis/modeling is central to many scientific fields, such as
bioinformatics and natural language processing, in other  fields - such as social economics, psychology, education and environmental
science - researchers are focused on testing and evaluating {\em causal hypotheses}. While the distinction between causal
and predictive analysis has been recognized, the conflation between the two is common.

Causal inference has been studied extensively in statistics and
computer science \cite{Fisher1935design,Rubin2005,holland1986statistics,PearlBook2000,Spirtes:book01}.
Many tools perform causal inference
 using statistical software such as SAS, SPSS, or R project. However, these toolkits do not scale to large datasets.  Furthermore, in many of the most interesting Big Data settings, the data
is highly relational (e.g, social networks, biological networks, sensor networks and more) and likely to pour into SQL systems. There is a rich ecosystem of tools and organizational requirements that encourage this. Transferring
data from DBMS to statistical softwares or connecting these
softwares to DBMS can be error prone, difficult, time consuming and inefficient. For these
cases, it would be helpful to push statistical methods for causal inference into the DBMS.

Both predictive and causal analysis are needed to generate
and test theories,  policy and decision making and to evaluate hypotheses, yet each plays a different role in doing so. In fact, performing predictive analysis to address questions that are causal in nature could lead to a flood of false discovery claims. In many cases, researchers who want to discover causality from data analysis settle for predictive analysis either because they think it is causal or lack of available alternatives.

This work introduces \GSQL,\footnote{ The prefix Zali refers to
  al-Ghzali (1058-1111), a medieval Persian philosopher. It is known
  that David Hume (1711-1776), a Scottish philosopher, who gave the
  first explicit definition of causation in terms of counterfactuals,
  was heavily influenced by al-Ghzali's conception of causality
  \cite{shalizi2013advanced}.}  a SQL-based framework for drawing
causal inference that circumvents the scalability issue with the
existing tools.  ZaliQL supports state-of-the-art methods for causal
inference and runs at scale within a database engine.  We show how to
express the existing advanced causal inference methods in SQL, and
develop a series of optimization techniques allowing our system to
scale to billions of records. We evaluate our system on a real
dataset.  Before describing the contributions of this paper, we
illustrate causal inference on the following real example.

\begin{table*}[!htb] \scriptsize
    \begin{subtable}{.5\linewidth}
      \centering

       \begin{tabular}[t]{|l|l|}
  \hline
  \bf{Attribute}   & \bf{Description}  \\  \hline
  FlightDate & Flight date   \\ \hline
  UniqueCarrier	& Unique carrier code   \\ \hline
  OriginAirportID & 	Origin airport ID  \\ \hline
  CRSDepTime & Scheduled departure time  \\ \hline
  DepTime & Actual departure time \\ \hline
    & difference in minutes between   \\
  DepDelayMinutes & scheduled and actual departure \\
   & time. Early departures set to 0  \\ \hline
  LateAircraftDelay & Late aircraft delay, in minutes  \\ \hline
  SecurityDelay & Ssecurity delay, in minutes \\ \hline
  CarrierDelay & Carrier delay, in minutes  \\ \hline
  Cancelled & Binary indicator  \\
  \hline
\end{tabular}
     \caption{Flight dataset}
    \end{subtable}%
    \begin{subtable}{.5\linewidth}
      \centering

       \begin{tabular}[t]{|l|l|}
  \hline
  \bf{Attribute}   & \bf{Description}  \\  \hline
  Code  & Airport ID    \\ \hline
  Date	& Date of a repost   \\ \hline
  Time        & Time of a report  \\ \hline
  Visim        & Visibility in km  \\ \hline
  Tempm & Temperature in C$^{\circ}$ \\ \hline
  Wspdm         & Wind speed kph \\ \hline
  \ignore{Precipm         & Humidity \% \\ \hline}
  Pressurem         & Pressure in mBar  \\ \hline
  Precipm         & Precipitation in mm  \\ \hline
  \ignore{Rain & Binary indictor    \\ \hline
  Snow & Binary indictor \\ \hline}
  Tornado & Binary indictor \\ \hline
  Thunder & Binary indictor \\ \hline
  Hum & Humidity \% \\ \hline
  Dewpoint & De point in  C$^{\circ}$ \\ \hline
\end{tabular}
        \caption{Weather dataset}
    \end{subtable}
 \vspace{-0.1cm}   \caption{\bf{List of attributes from the flight(a) and weather(b)  datasets that are relevant to our analysis.}}
\label{tab:attlist}
\end{table*}
\vspace{-0.1cm}

\vspace{-.2cm}
\begin{example} \em \delay. \ \em   Flight delays pose a serious
and widespread problem in the United States and
 significantly strain on the national air travel system, costing society many billions of dollars each year \cite{ball2010total}.
 According to FAA statistics,\footnote{National Aviation Statistic \url{http://www.faa.gov/}}
   weather causes approximately 70\% of the delays in the US National Airspace System (NAS). The upsetting impact
   of weather conditions on aviation is well known, however quantifying
  the causal impact of different weather types on flight delays
at different airports is essential for evaluating
approaches to reduce these delays. Even though predictive analysis, in this context, might help make certain policies,
this problem is causal. We conduct this causal analysis as a running example through this paper. To this end,
we acquired flight departure details for all commercial flights within the US from 2000 to 2015
(105M entries) and integrated it with the relevant historical weather data (35M entries) (see Section \ref{sec:setup}).
These are relatively large data sets for causal inference that can not be handseled by the existing tools. Table \ref{tab:attlist} presents
the list of attributes from each data set that is relevant to our analysis.

\end{example}

When we make predictive analysis, whether we predict $\E[Y | X = x]$ or $\textrm{Pr}(Y | X = x)$ or
something more complicated, we essentially want to know the conditional distribution of
$Y$ given $X$. On the other hand, when we make a causal
analysis, we want to understand the distribution of $Y$, if the
usual mechanisms controlling $X$ were intervened and set to $x$.
In other words, in causal analysis we are interested in {\em  interventional} conditional
distribution, e.g.,  the distribution  obtained by (hypothetically)
enforcing $X = x$ uniformly over the population.  In causal analysis, the difficulty arises   from the fact that here the objective is to estimate (unobserved)
   {\em counterfactuals} from the (observed) {\em factual} premises.

   \begin{example} \em \delay \ (Cont.).  \label{ex:press} \em Suppose
     we want to explore the effect of low-pressure on flight departure
     delays. High pressure is generally associated with clear weather,
     while low-pressure is associated with unsettled weather, e.g.,
     cloudy, rainy, or snowy
     weather\ignore{\cite{weba2,barometricpressureheadache:article}}. Therefore,
     conducting any sort of predictive analysis identifies
     low-pressure as a predictor for flight delays. However,
     low-pressure does not have any causal impact on departure delay
     (low-pressure only requires longer takeoff distance)
     \cite{FAA08}.  That is, low-pressure is most highly a correlated
     attribute with flight delays, however ZaliQL found that other
     attributes such as thunder, low-visibility, high-wind-speed and
     snow have the largest causal effect on flight delays (see
     Sec. \ref{sec:endtoend}); this is confirmed by the results
     reported by the FAA and \cite{weather}.

\end{example}

% They defined a very simple model, where we
% want to conclude \dan{Lise: what's your comment here?}  if one
% variable, called ``treatment'' causes one particular output variable,
% called ``effect'', and have developed a rich set of techniques for
% that purpose.  The end goal of their analysis is to establish the
% average causal-treatment effect.  More recently, in the AI literature
% Pearl~\cite{pearl2010introduction,PearlBook2000} has extended this
% simple model by introducing {\em causal networks}, and developing a
% logical framework for reasoning about causality.  Their aim is to
% enable complex inference in a network of causal associations.  While
% the ultimate goal is the same, to check whether a particular input
% causes a particular outcome, the methods deployed are different from
% those in statistics.

\vspace{-0.1cm}

This paper describes novel techniques implementing and optimizing
state-of-the-art causal inference methods (reviewed in Section
\ref{subsec:causalitystatistics}) in relational databases.  We make
three contributions.  First, in Section \ref{sec:BasicTechniques} we
describe the basic relational implementation of the main causal
inference methods: matching and subclassification.  Second, in Section
\ref{sec:OptimizationTechniques} we describe a suite of optimization
techniques for subclassfication, both in the online and offline
setting.  Finally, in Section \ref{sec:exp} we conduct an extensive
empirical evaluation of \GSQL, our system that implements these
techniques, on real data from the U.S. DOT and Weather Underground
\cite{flightdata,Weatherdata}.

\ignore{
This paper makes the following specific contributions: we describe a
suite of techniques for expressing the existing advanced methods for
causal inference from observational data in SQL that run at scale
within a database engine (Section \ref{sec:BasicTechniques}). Note
that we do not claim any contribution the existing methods; we
introduce several optimization techniques that significantly speedup
causal inference, both in the online and offline setting (Section
\ref{sec:OptimizationTechniques}); we validate our system
experimentally, using real data from the U.S. DOT and Weather
Underground \cite{flightdata,Weatherdata}
}

\vspace{-2mm}

%% file: background.tex
\vspace{-.1cm}
\section{Background: Causality Inference in Statistics}
\label{subsec:causalitystatistics}

\ignore{
\dan{I'm inclined to remove this paragraph} \babak{I agree, the example is repeated later on} Causal inference aims to
establish a casual relationship between one variable called the {\em
  cause} or {\em treatment}, which has two possible values ($0$ and
$1$) and another variable, called the {\em effect} or {\em outcome}.
When the treatment value is $1$, we say that the treatment was
applied; otherwise, we say that the control treatment was applied.  A
causal experiment is done with respect to a population or collection
of {\em units}.  A causal relationship implies that if we intervene by
changing the treatment, then we should expect to see a change in the
outcome.  For example (referring to Table~\ref{tab:attlist}) a unit is
a flight, the treatment is Thunder (a binary attribute) and the
outcome is DepDelayMinutes (the delay in minutes): we are interested
in whether thunder causes delay, and how much.}

% \subsection{The Potential Outcome Framework} \label{sec:NRCM}
The basic causal model in statistics is called the Neyman-Rubin Causal
Model (NRCM). This framework views causal effect as comparisons
between {\em potential outcomes} defined on the same units. This
section describes the basic framework.
\ignore{
The basic causal model in statistics is called the Neyman-Rubin Causal
Model (NRCM). This framework views causal effect as comparisons
between {\em potential outcomes} defined on the same units. This
section describes the basic framework.}

\begin{table*}[t] \scriptsize
  \centering
  \begin{tabular}{|c|c|c|c|c|c|} \hline
    Unit & $T$           & $X$          & $Y(1)$                & $Y(0)$        &  $Y(1)-Y(0)$ \\
         & (Treatment)   & (Covariates) & (Treated outcome)     & (Control outcome)     & (Causal  Effect)\\
    \hline
    1   & $T_1$ & $X_1$ &  $Y_1(1)$ & $Y_1(0)$ & $Y_1(1) - Y_1(0)$ \\
    2   & $T_2$ & $X_2$ &  $Y_2(1)$ & $Y_2(0)$ & $Y_2(1) - Y_2(0)$ \\
$\ldots$&       &       &  & & \\
    N   & $T_N$ & $X_N$ &  $Y_N(1)$ & $Y_N(0)$ & $Y_N(1) - Y_N(0)$ \\ \hline
  \end{tabular} \vspace{0.3cm}
  \caption{\bf{The Neyman-Rubin Causal Model (NRCM).}}
  \label{fig:causal:inference}
\end{table*}

\paragraph*{\bf{Average Treatment Effect (ATE)}}
In the NRCM we are given a table $R(T,X,Y(0),Y(1))$ with
$N$ rows called {\em units}, indexed by $i=1 \ldots N$; see Table~
\ref{fig:causal:inference}.  The binary attribute $T$ is called {\em
  treatment assignment} ($T=1$ means the unit was treated; $T=0$ means
the unit was subjected to control); $X$ is a vector of attributes
called {\em covariates}, unaffected by treatment; and the two
attributes $Y(0), Y(1)$ represent {\em potential outcomes}: $Y(1)$ is
the outcome of the unit if it is exposed to the treatment and $Y(0)$
is the outcome when it is exposed to the control.  For any attribute
$Z$ we write $Z_i$ for the value of the $i$'s unit.  The effect caused
by the treatment for the $i$th unit, simply called the {\em treatment
  effect} for the $i$th unit, is defined as $Y_i(1)-Y_i(0)$.  The goal
of causal analysis is to compute the {\em average treatment effect
  (ATE)}:

\vspace{-.4cm}
\begin{align}
  \ate = \E[Y(1)-Y(0)] = \E[Y(1)] - \E[Y(0)] \ignore{=\frac{1}{N} \sum_i (Y_i(1)-Y_i(0))}  \label{eq:ate}
\end{align}

Throughout this paper $\E[Z]$ refers to the expected value of the
attribute $Z$ of an individual chosen at random from a large
population.  The population is unavailable to us, instead we have the
database which is typically a random sample of $N$ {\em units} from
that population.  Then $\E[Z]$ is estimated by the empirical expected
value, $\E[Z] = \sum_i Z_i/N$, where $Z_i$ is the attribute of the
$i$'th unit in the database.  In this paper we do not address the
sampling error problem, but we point out that the precision of the
estimator increases with the sample size.  Thus, a key goal of the
techniques discussed below is to ensure that expected values are
computed over sufficiently large subsets of the data.
\vspace{-.15cm}
\begin{example} \em \label{exa1} \delay \ (Cont.). Suppose we want to
  quantify the causal effect of thunder on flight departure delays. In
  this case the table is the spatio-temporal join of the flights and
  weather data in Table \ref{tab:attlist}, the treatment $T$ is the
  Thunder attribute, the outcome is
  DepDelayMinutes, while the covariates are all other attributes
  unaffected by Thunder, e.g. flights carrier, the origin airport,
  traffic, some other weather attributes such as temperature.
\end{example}

Somewhat surprisingly, the model assumes that {\em both} $Y_i(1)$ and
$Y_i(0)$ are available for each unit $i$.  For example, if the
treatment $T$ is Thunder and the outcome is DepDelayMinutes, then the
assumption is that we have {\em both} values DepDelayMinutes, when
thunder was present and when it was absent.  The inclusion of both
outcomes, factual and counterfactual, in the data model is considered
to be one of the key contributions of the NRCM.
Of course, in reality we have only one of these outcomes for each
unit, \eg, if there was a thunder during that flight then we know
$Y_i(1)$ but not $Y_i(0)$, and in this case we simply write $Y_i$ to
denote $Y_i(T_i)$.  This missing value prevents us from computing
$\ate$ using Eq.(\ref{eq:ate}), and is called the {\em fundamental
  problem of causal inference} \cite{Holland1986}.  Therefore, in
order to compute $\ate$, the statistics literature makes further
assumptions.

\paragraph*{\bf{Randomized Data}}
The strongest is the {\em independence assumption}, which states that
the treatment mechanism is independent of the potential outcomes,
i.e., $(Y(1), Y(0)) \bigCI T$. Then, it holds that
$\E[Y(1)]=\E[Y(1)|T=1]$ and similarly $\E[Y(0)]=\E[Y(0)|T=0]$ and we
have:\footnote{An additional assumption is actually needed, called
  \emph{Stable Unit Treatment Value Assumptions (SUTVA)}, which states
  that the outcome on one unit is not affected by the treatment of
  another unit, and the treatment attribute $T$ is binary.  We omit
  some details.}
% \footnote{This framework relies on an assumption that referred to
%   as \emph{Stable Unit Treatment Value Assumptions (SUTVA)}.  This
%   assumption states that: (a) the potential outcome of an individual
%   unit is not affected by the treatment assignment to the other units;
%   \ (b) there is only a single version of a treatment. Assumption (a)
%   ensures that there is no interference between the units and
%   potential outcomes are well-defined.  Assumption (b) is required to
%   make causal inference reasonable. Specifically, If this assumption
%   is violated (e.g., when patients are treated with different dosage
%   of a drug), then inference is not reasonable, because potential
%   outcomes of different units are not comparable.}
%
\vspace{-.1cm}
\begin{eqnarray}
  \ate &=& \E[Y(1)|T=1]-\E[Y(0)|T=0] \label{eq:indeass}
\end{eqnarray}
Each expectation above is easily computed from the data, for example
$\E[Y(1)|T=1] = \sum_{i: T_i=1} Y_i / N_1$ where $N_1$ is the number
of units with $T=1$.  The golden standard in causal analysis are {\em
  randomized experiments}, where treatments are assigned randomly to
the units to ensure independence; for example, in medical trials each
subject is randomly assigned the treatment or a placebo, which implies
independence.
\begin{figure*} \scriptsize
  \centering
  \begin{tabular}{|l|l|l|} \hline
    \bf{Name} &  $\delta(x_i,x_j)=$ & \bf{Comments} \\ \hline
Coarsened distance & 0 if $C(x_i)=C(x_j)$ & Where $C(x)$ is a function that coarsen \\
               & $\infty$ if $C(x_i) \neq C(x_j)$ & a vector of continues covariate~\cite{IacKinPor09}\\ \hline
Propensity score  & $|E(x_i) - E(x_j)|$ & where $E(x) = \textrm{Pr}(T = 1 | X=x)$\\
distance (PS)   &  & is the propensity score~\cite{Rubin1983b} \\ \hline
Mahalanobis Distance (MD) & $(x_i-x_j)'\Sigma^{-1} (x_i-x_j)$ & where $\Sigma=$ covariance matrix~\cite{Stuart10} \\ \hline
  \end{tabular}
  \caption{\bf{Distance Measures used in Matching}}
  \label{fig:metrics}
\end{figure*}

\paragraph*{\bf{Observational Data}}
In this paper, however, we are interested in causal analysis in {\em
  observational data}, where the mechanism used to assign treatments
to units is not known, and where independence fails in general.  For
example, thunders occur mostly in the summer, which is also high
travel season and therefore delays may be caused by the high traffic.
In that case $T$ and $Y(0)$ (the delay when thunder does not occur)
are correlated, since they are high in the summer and low in the
winter, and similarly for $T,Y(1)$.  The vast majority of datasets
available to analysts today are observational data, and this motivates
our interest in this case.  Here, the statistics literature makes the
following weaker assumption~\cite{Rubin1983b}:

\begin{description}
\item[Strong Ignorability:] Forall $x$ the following hold: \newline
  (1) Unconfoundedness $({Y(0), Y(1) \bigCI T} | X=x)$ and \newline
  (2) Overlap $0 < \textrm{Pr}(T = 1 | X=x) < 1$
\end{description}

The first part, {\em unconfoundedness}, means that, if we partition
the data by the values of the covariate attributes $X=x$, then, within
each group, the treatment assignment and the potential outcomes are independent; then we can
estimate $\ate$ by computing Eq. \ref{eq:indeass} for each value of
the covariates $X=x$ (i.e., conditioning on $X$) and averaging.  The
second part, {\em overlap} is needed to ensure that the conditional
expectations $\E[Y(1)|T=1,X=x]$ and $\E[Y(0)|T=0,X=x]$ are well
defined.  For an illustration of unconfoundedness, suppose we restrict
to flights on dates with similar weather and traffic conditions, by a
fixed carrier, from a fixed airport, etc; then it is reasonable to
assume that $T$ and $Y(0)$ are independent and so are $T$ and $Y(1)$.
In order to satisfy unconfoundedness one has to collect sufficiently
many confounding attributes in $X$ about the data in order to break
any indirect correlations between the treatment and the outcome.

% Notice the need to weight by the
% size of the group, $N_x/N$: some matching techniques discussed below
% ensure that all groups have the same size, and the formula simplifies
% to a standard average.

However, once we include sufficiently many covariate attributes $X$
(as we should) then the data becomes sparse, and many groups $X=x$ are
either empty or have a very small number of items.  For example if a
group has only treated units, then the overlap condition fails, in
other words the conditional expectation $\E[Y(0)|X=x,T=0]$ is
undefined.
% The problem with this approach is that in general the data is
% very sparse: most values of $X=x$ have only one item (thus violating the
% overlap condition) or have only very few items (thus causing the
% empirical estimate of $\E[-]$ to have low precision).
In our example, if the covariate attributes include OriginAirportID,
UniqueCarrier,Tempm ,Wspdm and Precipm, because in that case all units
within a group will have the same value of Thunder.  In general the
strong ignorability assumption is not sufficient to estimate $\ate$ on
observational data.

\paragraph*{Perfect Balancing}
The solution adopted in statistics is to increase the size of the
groups, while ensuring that strong ignorability holds {\em within each
  group}.  In other words, instead of grouping by the values of the
covariates $X$, one groups by the values of some function on the
covariates $B(X)$.  We say that the groups are strongly ignorable if
the strong ignorability condition holds within each group: 
\begin{description}
\item[Strong Ignorability in Groups:] Forall $b$ the following holds:
  \newline (1) Unconfoundedness $({Y(0), Y(1) \bigCI T} | B(X)=b)$ and
  \newline (2) Overlap $0 < \textrm{Pr}(T = 1 | B(X)=b) < 1$
\end{description}
Rosenbaum and Rubin~\cite{Rubin1983b} gave an elegant characterization
of the functions $B$ that define strongly ignorable groups, which we
review here.  We say that the groups are {\em perfectly balanced}, or
that $B$ is a {\em balancing score} if $X$ and $T$ are independent in
each group, \ie $({X \bigCI T} | B(X)=b)$ for all values $b$.
Equivalently: \vspace{-0.2cm}
\begin{description}
\item[Perfect Balanced Groups:] Within each group $b$, the
  distribution of the covariate attributes of the treated units is the
  same as the distribution of the control units:

\end{description}
{
 \begin{align}
 \forall x:  & \textrm{Pr}(X = x | T = 1, B(X)=b) =
         \textrm{Pr}(X = x | T = 0,B(X)=b) \hfill \label{eq:bs}
  \end{align}}
 \begin{theorem}~\cite[Th.3]{Rubin1983b} \label{the:sib} If the
  treatment assignment is strongly ignorable, and $B$ defines
  perfectly balanced groups, then the treatment assignment is strongly
  ignorable within each group.
\end{theorem}

\begin{proof}
  The overlap part of the theorem is trivial, we show
  unconfoundendess.  Abbreviating $Y = (Y(0),Y(1))$, we need to prove:
  if (a) $({Y \bigCI T} | X=x)$ and (b) $({X \bigCI T} | B(X)=b)$,
  then\footnote{This is precisely the Decomposition Axiom in
    graphoids~\cite[the. 1]{PearlBook1998}; see ~\cite{DBLP:journals/ipl/GyssensNG14} for a
    discussion.}  $({Y \bigCI T} | B(X)=b)$.  (a) implies
  $\E[T | Y=y, X=x] = \E[T | X=x]$ and also
  $\E[T | Y=y, X=x, B=b] = \E[T | X=x, B=b]$ since $B$ is a function
  of $X$; (b) implies $\E[T | X=x] = \E[T | X=x, B=b] = \E[T | B=b]$.
  Therefore,
  $\E[T | Y=y, B=b] = \E_x[\E[T | Y=y,X=x,B=b]] = \E_x[\E[T | X=x,
  B=b]] = \E[T | B=b]$ proving the theorem.
\end{proof}

If the treatment assignment is strongly ignorable within each group,
then we can compute $\ate$ by computing the expectations of
Eq.~\ref{eq:ate} in each group, then taking the average (weighted by
the group probability):

\vspace{-.3cm}
{
\begin{eqnarray}
\lefteqn{\ate = \E[Y(1)-Y(0)] = \E_b[\E[Y(1)-Y(0)| B(X)=b]]} \label{eq:oatt} \\
  &=& \E_b[\E[Y(1)|B(X)=b]]-\E_b[\E[Y(0)|B(X)=b]] \nonumber \\
  &=& \E_b[\E[Y(1)|T=1,B(X)=b]]-\E_b[\E[Y(0)|T=0,B(X)=b]] \nonumber
\end{eqnarray}
}
Thus, one approach to compute causal effect in observational data is
to use group the items into balanced groups, using a balancing fuction
$B$.  Then, $\ate$ can be estimated using the formula above, which can
be translated into a straightforward SQL query using simple
selections, group-by, and aggregates over the relation $R$.  This
method is called {\em subclassification} in statistics.  The main
problem in subclassification is finding a good balancing score $B$.
Rosenbaum and Rubin proved that the best balancing score is the
function $E(x) = \textrm{Pr}(T = 1 | X=x)$, called {\em propensity
  score}.  However, in practice the propensity score $E$ is not
available directly, instead needs to be learned from the data using
logistic regression, and this leads to several
problems~\cite{king15}.  When no good balancing function can be found,
a related method is used, called matching.

\paragraph*{\bf{Matching}} We briefly describe matching
following~\cite{Rubin1983b}.  Consider some balancing score $B(X)$
(for example the propensity score).  One way to estimate the quantity
in Eq.~\ref{eq:oatt} is as follows.  First randomly sample the value
$b$, then sample one treated unit, and one control unit with $B(X)=b$.
This results in a set of treated units $i_1, i_2, \ldots, i_m$ and a
matching set of control units $j_1, j_2, \ldots, j_m$: then the
difference of their average outcome
$\sum_k (Y_{i_k}(1) - Y_{j_k}(0))/m$ is an unbiased estimator of
$\E_b[\E[Y(1)-Y(0)| B(X)=b]]$ and, hence, of $\ate$.  Notice that
there is no need to weight by the group size, because the ratio of
treated/untreated units is the same in each group.  Generally, the
matching technique computes a subset of units consisting of all
treated units and, for each treated unit, a randomly chosen sample of
fixed size of control units with the same value of balancing score.

% The motivation for this technique is to reduce the cost associated
% with measuring outcome. For example, in clinical studies typically a
% few patients that are treated with a drug exists that can be matched
% with potentially a very large number of control units. Measuring an
% outcome for all possible control units might become infeasible. Notice
% that matching or grouping based on a balancing score can be performed
% before measuring an outcome (see \cite{Rubin1983b} and references
% therein).

Historically, matching  predated
subclassification, and can be done even when no good balancing score
is available.
%
% when no good
% Ideally, treated and control units would exactly matched on some
% balancing score $B(X)$, so that the {\em sample} distribution of
% covariates $X$ would be identical in two groups. In practice, however,
% exact matches and thus groups even on a scalar balancing score are
% often impossible to obtain so methods which seeks approximate matches
% must be used \cite{Rubin1983b}.
%
The idea is to match each treated unit with one or multiple control
units with ``close'' values of the covariate attributes $X$, where
closeness is defined using some distance function $\delta(x_i,x_j)$
between the covariate values of two units $i$ and $j$.  The most
commonly used distance functions are listed in
Fig.~\ref{fig:metrics}. \footnote{The distance functions in Table
  \ref{fig:metrics} are {\em semi or pseudo-metrics}. That is they are
  symmetric; they satisfy triangle inequality and $x_i=x_j$ implies
  $\delta(x_i,x_j)=0$, but the converse does not hold.} The efficacy
of a matching method is evaluated by measuring {\em degree of
  imbalance} i.e., the differences between the distribution of
covariates in two groups in the matched subset. Since there is no
generic metric to compare two distributions, measures such as mean,
skewness, quantile and multivariate histogram are used for this
purpose.  A rule of thumb is to
evaluate different distance metrics and matching methods until a
well-balance matched subset with a reasonable size obtained.

\ignore{
\paragraph*{Matching}
Instead of grouping by some balancing score $B(x)$, several techniques
in the statistics literature compute the groups using {\em
  matching}~\cite{Rubin1974}, which computes a subset of units
consisting of all treated units and, for each treated unit, a random
sample of fixed size of control units with ``close'' values of the
covariate attributes, where closeness is defined using one of the
distance measures $\delta(i,j)$ between units listed in
Fig.~\ref{fig:metrics}.  To see the intuition behind matching, suppose
the data is dense, and we can match each treated unit with one control
unit with the same values of the covariate attributes.  Then, each
group (defined by the value of the covariates) is perfectly balanced,
and moreover the ratio of treated/control units is the same in all
groups, hence we do not care about the group sizes in
Eq.\ref{eq:oatt}.  In general, matching relaxes this principle, by
matching each treated unit with a fixed number of control units (still
ensuring a constant ratio) and relaxing the equality condition to
closeness according to some measure.  Thus, the main computational
challenge in causal inference consists of matching: we describe it in
detail in Sec.~\ref{?????}.  \dan{I don't like how this paragraph come out}}

\paragraph*{\bf{Summary}}
The goal of causal analysis is to compute $\ate$ (Eq.~\ref{eq:ate})
and the main challenge is that each record misses one of the outcomes,
$Y(1)$ or $Y(0)$.  A precondition to overcome is to ensure strong
ignorability, by collecting sufficiently many covariate attributes
$X$.  For the modern data analyst this often means integrating the
data with many other data sources, to have as much information
available as possible about each unit.  One caveat is that one should
not include attributes that are themselves affected by the treatment;
the principled method for choosing the covariates is based on
graphical models~\cite{de2011covariate}.  Once the data is properly
prepared, the main challenge in causal analysis  \footnote{There are
  some model-based alternatives to matching e.g., covariates
  adjustment on random samples.  However, matching have several nice
  properties that makes it more appealing in practice (see,
  \cite{Rubin1983b}).} is {\em matching} data records such as to
ensure that the distribution of the covariates attributes of the
treated and untreated units are as close as possible
(Eq.(\ref{eq:bs})).  This will be the focus of the rest of our paper.
Once matching is performed, $\ate$ can be computed using
Eq.(\ref{eq:oatt}).  Thus, the main computational challenge in causal
analysis is the matching phase, and this paper describes scalable
techniques for performing matching in a relational database system.

%% file: btech.tex
\section{Basic Techniques}
\label{sec:BasicTechniques}

% \subsection{Declarative Specification of the Matching Methods}
% \label{sec:BasicTechniquesdef}

In this section we review the matching and subclassification
techniques used in causal inference and propose several relational
encodings, discussing their pros and cons.  Historically, matching was
introduced before subclassification, so we present them in this order.
Subclassification is the dominant technique in use today: we will
discuss optimizations for subclassification in the next section.

We consider a single relation $\rel(ID,T,\cv,Y)$,
where $ID$ is an integer-valued primary-key, $T$ and $\cv$
respectively denote the treatment and covariate attributes as
described in the NRCM (cf. Section \ref{subsec:causalitystatistics}),
and $Y$ represent the available outcome, i.e., $Y=Y(z)$ for iff $T=z$. For each matching method, we define a view over $\rel$ such that
materializing the extension of the view over any instance of $\rel$
computes a corresponding matched subset of the instance.

\vspace{0.3cm}

\begin{figure}
  \centering
\begin{alltt} 
CREATE VIEW \(\nnmwr\) AS
SELECT *
   FROM \(\rel\) AS control,\(\rel\) AS treated
WHERE control.T=0  
  AND treated.T=1
  AND \(\delta\)(treated\(.\cv\), control\(.\cv)\) < \(caliper\)
  AND (SELECT count(*)
     FROM \(\rel\) AS z
     WHERE z.T=0
       AND \(\delta\)(treated\(.\cv\),z\(.\cv)\) < \(\delta\)(treated\(.\cv\),control\(.\cv)\) \(\leq k\))
\end{alltt} \vspace{-.5cm}
\center{\bf (a)  Anti-join based}
\vspace{-.4cm}
\begin{alltt} 
CREATE VIEW \(\nnmwr\) AS
WITH potential_matches AS
  (SELECT treated.ID AS tID,
          control.ID AS cID,
         \(\delta\)(treated\(.\cv\),control\(.\cv)\)  AS distance
   FROM \(\rel\) AS control,
        \(\rel\) AS treated
   WHERE control.T=0 
     AND treated.T=1
     AND \(\delta\)(treated\(.\cv\),control\(.\cv)\) < \(caliper\)),
      ranked_potential_matches AS
  (SELECT *, ROW_NUMBER() OVER (PARTITION BY tID
                       ORDER BY distance) AS order
   FROM potential_matches)
SELECT *
FROM ranked_potential_matches
WHERE order \(\leq k\)
\end{alltt}
\vspace{-.5cm}
\center{\bf (b) Window function based}
  \caption{\bf SQL implementation of NNMWR.}\label{fig:nnmwr}
\end{figure}

\subsection{Nearest Neighbor Matching}
\label{sec:nnm}
The most common matching method is that of $k:1$ nearest neighbor
matching (NNM) \cite{Rubin1983b,ho2005,Stuart10}. This method selects
the $k$ nearest \ignore{(wrt. one of the distance metrics discussed in
  Section \ref{sec:matching})} control matches for each treated
unit and can be done with or without replacement; we denote
them respectively by NNMWR and NNMNR.  In the former case, a control
unit can be used more than once as a match, while in the latter case
it is considered only once. Matching with replacement can often
decrease bias because controls that look similar to the treated units
can be used multiple times.  This method is helpful in settings where
there are few control units available. However, since control units
are no longer independent, complex inference is required to estimate
the causal effect \cite{dehejia99}. In practice, matching is usually
performed without replacement.  Notice that NNM faces the risk of bad
matches if the closest neighbor is far away. This issue can be
resolved by imposing a tolerance level on the maximum distance, known
as the {\em caliper} (see e.g., \cite{lunt2014selecting}). There are some rules of thumb for choosing the
calipers (see e.g., \cite{lunt2014selecting}).

{\bf NNM With Replacement}
We propose two alternative ways for computing  NNMWR in SQL, shown in
Figure \ref{fig:nnmwr}. In Figure \ref{fig:nnmwr}(a), each treated unit is joined with $k$ closest control units that are closer than the caliper. In this solution, nearest control units are identified by means of an anti-join. \ignore{Note that when $k=1$ the aggregate expression may be replaced by NOT EXISTS.}  In Figure \ref{fig:nnmwr}(b), all potential matches and their distances are identified by
joining the treated with the control units that are closer than
the caliper. Then, this set is sorted into ascending order of
distances.  In addition, the order of each row in the sorted set is identified
using the window function {\verb|ROW_NUMBER|}. Finally, all units with the order of less than or equal to $k$ are selected as the matched units.

The ani-join based statement requires a three-way join. \ignore{However, in a particular case when
  $k=1$ and the caliper is zero, this solution can becomes linear for
  matching based on propensity score.} The window function based
solution has a quadratic complexity. It requires a nested-loop to
perform a spatial-join and a window aggregate to impose
minimality. Note that window functions are typically implemented in
DBMS using a sort algorithm, and even more efficient algorithms have
been recently proposed~\cite{Neumann15}.

\begin{figure}
  \centering
\begin{alltt} 
CREATE VIEW \(\nnmnr\)
AS WITH potential_matches AS
  (SELECT treated.ID AS tID, 
          control.ID AS cID,
          \(\delta(treated.\cv,control.\cv)\)  AS distance
   FROM \(\rel\) AS control, 
        \(\rel\) AS treated
   WHERE control.T=0 AND treated.T=1
     AND \(\delta(treated.\cv,control.\cv)\) < \(caliper\))),
            ordered_potential_matches AS
  (SELECT *, ROW_NUMBER() over (ORDER BY distance) AS order
   FROM potential_matches)
SELECT *
FROM ordered_potential_matches AS rp
WHERE NOT EXISTS
    (SELECT *
     FROM ordered_potential_matches AS z
     WHERE z.order < rp.order AND z.cID=rp.cID)
  AND (SELECT count(*)
     FROM ordered_potential_matches AS rp
     WHERE z.order < rp.order 
           AND z.tID=rp.tID)\( \leq k\)
\end{alltt} 
  \caption{\bf SQL implementation of NNMNR}\label{fig:nnmnr}
\end{figure}

{\bf NNM Without Replacement} Expressing NNMNR in a declarative manner
can be complicated. In fact, this method aims
to minimize the average absolute distance between matched units and can performed in either greedy or optimal manner. The latter is called {\em optimal matching}
\cite{Rosenbaum93}. Before we describe our proposed SQL implementation for NNMWR, we prove
that optimal matching is not expressible in SQL: this justifies
focusing on approximate matches.  For our inexpressibility result,
notice that in the special case when $k=1$ NNMWR is the {\em weighted
  bipartite graph matching problem (WBGM)}, which is defined as
follows: given a bipartite graph $G=(V,E)$ and a weight function
$w: E \rightarrow \mathbb{R}_{>0}$, find a set of vertex-disjoint
edges $M \subseteq E$ such that $M$ minimise the total weight
$w(M) = \sum_{e \in M} w(e)$.  The exact complexity of this problem
is unknown (see, e.g. \cite{Avis83}), however we prove a \NLOGSPACE\
lower bound:

\begin{proposition} \label{pro:om}
Computing  maximum weight matching for  weighted bipartite graphs is hard for \NLOGSPACE.
\end{proposition}

\begin{proof} The following {\em Graph Reachability Problem} is known
  to be \NLOGSPACE\ complete: given a directed graph $G(V,E)$ and two
  nodes $s,t$, check if there exists a path from $s$ to $t$.  We prove
  a reduction from graph reachability to the {\em bipartite perfect
    matching problem} which is a special case of optimal WBGM. For
  that we construct the graph $G'$ with $V= V \cup V'$ where, $V'$ is
  a copy of $V$ with primed labels and
  $E'= \{(x,x')| \forall x \in V -\{s,t\} \} \cup \{(x,y')| \forall
  (x,y) \in E\} \cup \{(t,s')\}$.
  Notice that the subset $\setof{(x,x')}{x \in V} \subseteq E$ is
  almost a perfect matching, except that it misses the nodes $s, t'$.
  We prove: there exists a path from $s$ to $t$ in $G$ iff $G'$ has a
  perfect matching.  First assume $P=s, x_1, x_2, \ldots, x_m, t$ is a
  path in $G$. Then the following forms a perfect matching in $G'$:
  $M=\set{(s,x_1'), (x_1,x_2'), \ldots, (x_{m},t'), (t,s')} \cup
  \setof{(y,y')}{y \not\in \set{s,x_1, \ldots, x_m,t}}$.
  Conversely, assume $G'$ has a perfect matching. Write
  $f : V \rightarrow V'$ the corresponding bijection, i.e. every $x$
  is matched to $y'=f(x)$.  Denoting the nodes in $V$ as
  $V=\set{x_1, \ldots, x_n}$, we construct inductively the following
  sequence: $x_{i_1}' = f(s)$, $x_{i_2}' = f(x_{i_3})$, \ldots,
  $x_{i_{k+1}}' = f(x_{i_k})$.  Then $i_1, i_2, \ldots$ are distinct
  (since $f$ is a matching), hence this sequence must eventually reach
  $t'$: $t' = f(x_{i_m})$.  Then
  $s,x_{i_1},x_{i_2}, \ldots, x_{i_m}, t$ forms a path from $s$ to $t$
  in $G$.  This completes the proof.
\end{proof}

The proposition implies that optimal matching is not expressible in
SQL without the use of recursion.  Optimal matching can be solved in
\PTIME using, for example, the {\em Hungarian} algorithm, which, in
theory, could be expressed using recursion in SQL.  However, optimal
matching is rarely used in practice and, in fact, it is known that it
does not in general perform any better than the greedy NNM (discussed
next) in terms of reducing degree of covariate imbalance
\cite{Rosenbaum93}.   For that reason, we did not implement optimal
matching in our system.

1:1 NNMWR can be approximated with a simple greedy algorithm that
sorts all edges of the underlying graph in ascending order of weights
and iterates through this sorted list, marking edges as ``matched"
while maintaining the one-to-one invariant. This algorithm can return
a maximal matching that is at least $\frac{1}{2}$-optimal
\cite{Avis83}. Figure \ref{fig:nnmnr} adopts this greedy algorithm to
express $1:k$ NNMWR in SQL.  This algorithm is very similar to that of
NNMWR in Figure \ref{fig:nnmwr}(b), with the main difference that in
the matching step it imposes the restriction that a control unit is
matched with a treated unit only if it is not not already matched with
another treated with a lower order.  This solution also has a
quadratic complexity.

{\bf Choosing the distance function} We briefly discuss now the choice
of the distance function $\delta$ in NNM (see Fig.~\ref{fig:metrics}).
The propensity score distance is by far the most prominent metric in
NNM.  However, it has been the subject of some recent criticisms
\cite{king15}.  It has been shown that, unlike other matching methods,
in propensity score matching the imbalance reduction is only
guaranteed across the spectrum of all samples. In observational
settings, we typically have only one sample, so other matching methods
dominate propensity score matching \cite{king15}. An alternative is to use the mahalanobis distance.  This has been
shown to exhibit some odd behavior when covariates are not normally
distributed, when there are relatively large number of covariates, or
there are dichotomous covariates
\cite{rosenbaum2002observational}. Therefore, this method has a
limited practical applicability.

We should mention that there is huge literature in the database
community on finding the nearest neighbor. In fact this type of
queries are subject of an active research and development efforts in
the context of spatial-databases (see, e.g.,
\cite{obe2015postgis}). Our work is different from these efforts in that: 1) much of the work in this area has focused on finding sub-linear algorithm for identifying nearest neighbors of a single data item (e.g., by using spatial-index). In contrast, in our setting
we need to find all nearest neighbors, which is by necessity quadratic; 2) these works resulted in specialized algorithm,
implemented in general purposed languages. In contrast, we focus on finding a representation in SQL, in order to integrate causal
 analysis with other data analytic operations.  \ignore{ Recent developments in this context are
applicable to our problem.  However, due to the enumerated
shortcomings of NNM based on Mahalanobis and propensity score
distance, this paper focuses on other matching methods.  \dan{which
  other matching methods?  also, if we don't use the current matching
  methods then why do we mention them?  and if research on NN is
  relevant, then what is novel here?}}

\begin{figure}
\begin{alltt} 
CREATE VIEW \(\sbc\) AS
(WITH tmp0 AS
  (SELECT *. ntile(\(n\)) over w subclass,
   FROM \(\rel\) window w AS (ORDER BY ps))
SELECT ID, T, \(\ccv\), Y, subclass,
             max(T) over w maxT, 
             min(T) over w minT
FROM tmp0  window w AS (PARTITION BY BY subclass)
WHERE maxT!=minT)
\end{alltt}
\vspace{-0.3cm}
  \caption{\bf{SQL implementation of subclassification based on the
      propensity score.}}\label{fig:subpr}
\end{figure}

\vspace{-.2cm}

\subsection{Subclassification}
\label{sec:sub}
It is easy to see that NNM does not necessarily use all the data,
meaning that many control units despite being in the range of a
treatment unit are discarded.  In subclassification, the aim is to
form subclasses for which, the distribution of covariates for the
treated and control groups are as similar as possible. The use of
subclassification for matching can be traced back to
\cite{cochran1968effectiveness}, which examined this method on a
single covariate (age), investigating the relationship between lung
cancer and smoking. It is shown that using just five subclasses based
on univariate continues covariates or propensity score removes over
$90\%$ of covariates imbalance \cite{cochran1968effectiveness,rosenbaum1984reducing}.

{\bf Subclassification based on the propensity score} We describe the
SQL implementation of subclassification based on $n$ quintiles of the
propensity score in Figure \ref{fig:subpr}.  We assumed that $\rel$
includes another attribute $ps$ for the propensity score of each unit;
the value of $ps$ needs to be learned from the data, using logistic
regression \cite{Rubin1983b}.  The SQL query seeks to partition
the units into five subclasses with propensity scores as equal as
possible using the window function {\verb|ntile|} and ensure the overlap within each subclass. \ignore{. Finally,
subclasses for which the maximum and minimum of the treatment are not
equal retained (the discarded subclasses do not enjoy the overlap
assumption).} This solution has the order of $nlog(n)$ if the window
function computed using a sort algorithm.

%%%%  we haven't defined univariates; i don't know what they are
% Subclassification based on univariate continues covariates is a
% particular case of coarsen exact matching, which is discussed in the
% next section.

\begin{figure}
\begin{alltt} 
CREATE VIEW \(\crele\) AS
(SELECT *, (CASE WHEN \(x\sb{1}<c\sb{1}\)  THEN \(1 \ldots\)
                 WHEN \(x\sb{1}>c\sb{(k\sb{1}-1)}\) THEN \(k\sb{1}\)) AS \(cx\sb{1}\),
              . . .
           (CASE WHEN \(x\sb{n}<c\sb{n}\)  THEN \(1 \ldots\)
                 WHEN \(x\sb{n}>c\sb{(k\sb{n}-1)}\) THEN \(k\sb{n}\)) AS \(cx\sb{n}\)
from \(\rel\))
\end{alltt} \vspace{-.2cm} \hspace{1.7cm}
{\bf{(a) Coarsening wrt. a set of prespecified cutpoints}}
\vspace{-.3cm}
\begin{alltt} 
CREATE VIEW \(\cem\) AS
SELECT ID, T, \(\ccv\), Y, subclass
FROM
  (SELECT *,
          max(ID) OVER w AS subclass, 
          max(T) OVER w AS minT,
          min(T) OVER w AS maxT
   FROM \(\crele\)
   WINDOW w (PARTITION BY \(\ccv\)))
WHERE mint!=maxt
\end{alltt}\vspace{-.2cm} \hspace{3cm}
\bf{(b) Window function based}
\vspace{-.3cm}
\begin{alltt} 
CREATE VIEW \(\cem\) AS
WITH subclasses AS
  (SELECT *,
          max(ID) OVER w subclass, 
          max(T) OVER w AS minT,
          min(T) OVER w AS maxT
   FROM \(\crele\)
   Group by \(\ccv\))
SELECT ID, T, \(\ccv\), Y, subclass
FROM subclasses,\(\crele\)
WHERE subclasses.\(\ccv\)=\(\crele\).\(\ccv\)  AND  minT!=maxT

\end{alltt}\vspace{-.2cm}\hspace{3.5cm}
 \bf{(c) Group-by based}
\vspace{-.1cm}
  \caption{\bf{SQL implementation of CEM.}}\label{fig:cem}
\end{figure}

% \subsection{Coarsen Exact Matching (CEM)}
% \label{sec:cem}

{\bf Coarsening Exact Matching (CEM)} This method as proposed
recently in \cite{IacKinPor09},  is a particular form of subclassification in which the vector of covariates $\cv$ is
coarsened according to a set of user-defined cutpoints or any
automatic discretization algorithm.  Then all units with similar
coarsened covariates values are placed in unique subclasses. All
subclasses with at least one treated and one control unit are retained
and the rest of units are discarded.  Intuitively, this is a group-by
operation, followed by eliminating all groups that have no treated, or
no control unit. \ignore{Subclassification based
  on univariate continues covariate (cf Section \ref{sec:sub}) can bee
  seen as a particular form of CEM.}

For each attribute $x_i \in \cv$, we assume a set of cutpoints
$c_i=\{c_{1} \ldots c_{(k_i-1)}\}$ is given, which can be used to
coarsen $x_i$ into $k_i$ buckets. The view $\crele$, shown in Figure
\ref{fig:cem}(a), defines extra attributes
$\ccv=\{\ccvi_1 \ldots \ccvi_n\}$, where $\ccvi_i$ is the coarsened
version of $x_i$. Two alternative SQL implementations of CEM are
represented in Figure \ref{fig:cem}(b) and (c).  The central idea in
both implementations is to partition the units based on the coarsened
covariates and discard those partitions that do not enjoy the overlap assumption. \ignore{Then, each partition in which the maximum and minimum of
the unit treatments are equal are discarded. This is because they do
not enjoy the overlap assumption.} Note that the maximum of unit
$ID$s in each partition is computed and used as its unique
identifier. The window function based solution has the order of
$nlog(n)$, if the window aggregate computed using a sort
algorithm. The group-by based solution can becomes linear if the join
is performed by a hash-join.

Several benefits of CEM has been proposed in \cite{IacKinPor09}. For
instance, unlike other approaches, the degree of imbalance is bounded
by the user (through choosing proper cut-points for covariates
coarsening), therefore the laborious process of matching and checking
for balance is no longer needed.  More importantly, this approach
meets the {\em congruous principle}, which assert that there should be
a congruity between analysis space and data space. Notice that
Mahalanobis distance and propensity score, project a vector from
multidimensional space into a scalar value. It has been argued that
methods violating the congruous principle may lead to less robust
inference with sub-optimal and highly counterintuitive properties
\cite{IacKinPor09}. Therefore, the reminder of the paper focuses on
developing optimization techniques to speed up the computation of
CEM.

\vspace{-2mm} 

%% file: optimization.tex
\section{Optimization Techniques}
\label{sec:OptimizationTechniques}

\ignore{
This section introduces
  several optimization techniques that significantly speedup  CEM, both in the online and offline setting.
\ignore{ Throughout this section, we assume a linear cost model for  CEM. That is, we
 assume that the cost of computing an extension of $\cem$ (cf. Figure \ref{fig:cem}(b)) is proportional to the number of  input rows. Our experimental results (Section \ref{sec:sct})  justify this assumption.}
}

\subsection{CEM on Base Relations}
\label{sec:baserel}

All toolkits developed for causal inference
assume that the input is a single table. However, in the real world, data is normalized, and stored in multiple tables
connected by key/forgien-keys. Thus, an analyst typically integrates tables
to construct a single table that contains all intergradients needed to conduct causal analyses. For instance, in  the \delay \ example, the treatment and part of the covariates
are stored in the weather dataset; the outcome and  rest of the covariates are stored in the
flight dataset. The fact that  data is scattered across multiple tables raises the question of whether we
can  push the matching methods to normalized databases. If so, we must question  whether we can take advantage of this property
to optimize the cost of  performing matching.

Integrating tables is inevitable for propensity score matching. For example, suppose we have two tables
$R(T,x,y,z)$ and $S(x,u,v)$. To estimate the propensity scores of each unit in $R \bowtie S$, we may fit
a logistic-regression between $T$ and  the covariates $y,z,u,$ and $v$. This may require computing the expression $w_1*y + w_2*z + w_3*u + w_4*v$ and then applying the logit function to it.  While the weights of the expression may be learned without joining the tables, using techniques such as \cite{schleich2016learning}, the integrated table is required to impute the leaned model with the covariate values of each unit to estimate its propensity score. In contrast,  CEM can be pushed
to the normalized databases. For example, $\cem(R \bowtie S)$ is equivalent to $\cem(\cem(R)\bowtie S)$. To see this, note that all subclasses discarded by performing CEM on $R$ do not satisfy the overlap assumption. It is clear that joining these subclasses with $S$, forms new subclasses that still fail to satisfy the overlap assumption and must be discarded. In the following, we formally state this property.

Let $D$ be a standard relational schema with $k$ relations
$\rel_1 \ldots \rel_k$, for some constant $k \geq 1$. The relation $\rel_i$
has the following attributes:  a  primary-key $ID_i $;  a foreign-key $FID_i $;   a vector of observed attributes
 $A_i$.  Without loss of generality, assume the treatment variable, $T$, is in relation $\rel_1$.  Let
  $\ccv_{\rel_i} \subseteq {A}_i$ be a vector of coarsened covariates from the relation $\rel_i$ that is associated with $T$.
    Further, assume relations are joined in the increasing order of indices.

\begin{proposition} \label{pro:push}
Given an instance of $D$, it holds that: $\cem(\rel_1  \bowtie
\ldots  \bowtie\rel_{k})= \cem( \ldots \cem(\cem(R_1) \bowtie   R_2) \ldots  \bowtie R_{k}) $.
\end{proposition}

Proposition \ref{pro:push} shows that CEM can be pushed to normalized databases.  In the worst case, the cost of pushing CEM can be $k-2$ times higher than performing CEM
on the integrated table. This happens when relations have a one-to-one  relationship  and
CEM retains all the input data. However, in practice the relations typically have a  many-to-one
relationship. Moreover, the size of the matched subset is much smaller than the input database. In the \delay \ example,
each row in the weather dataset is associated with many rows in the flight dataset. In addition, as we see
in Section \ref{sec:endtoend}, the size of the matched data is much smaller than the input data. In such settings,
pushing CEM down to the base relations can significantly reduce its cost.
\ignore{This significantly reduces the cost of matching if
tables have many-to-one relationships and CEM prunes considerable portion of data. In Section \ref{sec:opt}, we show the efficacy of applying this techniques on the \delay\ example.
\ignore{In the \dela \ example, each row in the weather dataset is associated with many rows in the flight dataset. In addition, as we see
in Section \ref{sec:endtoend}, the size of the matched data is much smaller than the input data.}}

\ignore{
In the worst case, the cost of pushing CEM can be $k-2$ times higher than performing CEM
on the integrated table. This happens when relations have a one-to-one  relationship  and
CEM retains all the input data. However, in practice the relations typically have a  many-to-one
relationship. Moreover, the size of the matched subset is much smaller than the input database. In the \delay \ example,
each row in the weather dataset is associated with many rows in the flight dataset. In addition, as we see
in Section \ref{sec:endtoend}, the size of the matched data is much smaller than the input data. In such settings,
pushing CEM down to the base relations can significantly reduce its cost. \ignore{
A simple  strategy would be to start with performing on the relation $\rel_1$that contains the treatment and join the result with a relation with the minimum size that can joined with, until all relations are processed.} \ignore{Section  \ref{sec:endtoend}  shows that this strategy significantly
reduces the cost of matching in the \delay \ example}
}

\vspace{-.1cm}
\subsection{Multiple Treatment Effect}

\label{sec:mte}

Matching methods are typically developed for estimating the causal effect of a single treatment
on an outcome. However, in practice one needs
to explore and quantify the causal effect of multiple treatments.
For instance, in the \delay \ example, the objective is to quantify and compare the causal effect
of different weather types on flight departure delays.

This section introduces online and offline techniques to speed up the computation of  CEM
for multiple treatments.  In the sequel, we consider the relational schema consists of a single relation $\rele(ID,\tre,\ccv,Y)$ (extends that of $\crele$ (cf. Section \ref{sec:sub}) to account for multiple treatments), where $\tre=T_1, \ldots, T_k$ is a vector of $k$ binary
treatments, each of which has a vector of coarsened covariates $\ccv_{T_i}$, with $\ccv= \bigcup_{i=1 \ldots k} \ccv_{T_i}$.
Now the view $\relei(ID,T_i,\ccv_{T_i},Y)$ over $\rele$ has the same schema
as $\crele$ (cf. Section \ref{sec:sub}). Therefore, the view $\cem$ (cf. Figure \ref{fig:cem}) is well-defined over  $\relei$.

\subsubsection{(online) Covariate Factoring}

A key observation for reducing the overall cost of performing CEM for multiple treatments is that
many covariates are shared between different treatments. For instance, flights carrier, origin airport, traffic and many weather attributes are shared between the treatments Thunder and LowVisibility. The central idea in {\em covariate factoring}  is to pre-process the input data
wrt.  the shared covariates between treatments and uses the result to perform CEM
for each individual treatment. This  significantly reduces the overall cost of CEM for all treatments,
if covariate factoring prunes a considerable portion of the input data.

\begin{figure}
\begin{alltt}
CREATE VIEW \(\prp\) AS
WITH tmp0 AS
  (SELECT *,
          max(ID) OVER w AS supersubclass,
          max(\(T\sb1\)) OVER w AS maxT\(\sb1\),..., maxT(\(T\sb{k'}\)) OVER w AS maxT\(\sb{k'}\),
          min(\(T\sb1\)) OVER w AS minT\(\sb1\),..., minT(\(T\sb{k'}\)) OVER w AS minT\(\sb{k'}\)
   FROM \(\rele\)
   WINDOW w (PARTITION BY \(\ccvin\)))
SELECT ID, \(\ccvu\), Y , supersubclass
FROM tmp0
WHERE max(\(T\sb1\))!=max(\(T\sb{1}\)) or ... or  max(\(T\sb{k'}\))!=max(\(T\sb{k'}\))
\end{alltt}  \hspace{4.5cm}
(a)\bf{Covariates factoring.}
\vspace{-.1cm}
\begin{alltt}
CREATE VIEW \(\mcem\) AS
WITH tmp0 AS
  (SELECT *,
          max(ID) OVER w    subclass,
          max(\(T\sb{i}\))  OVER w AS minT,
          max(\(T\sb{i}\))  OVER w AS maxT
   FROM \(\prpi\)
   WINDOW w (PARTITION BY supersubclass, \( \ccv\sb{T\sb{i}} \smallsetminus \ccvin  \)))
SELECT ID,\(T\sb{i}\),\(\ccv\sb{T\sb{i}}\), Y, subclass
FROM tmp0
WHERE minT!=maxT
\end{alltt}
 \hspace{5cm}
\bf{(b) Modified CEM}
\caption{\bf{CEM based on covariate factoring.}}\label{fig:cf}
\end{figure}

\ignore{
Let let $\ccvin$ and $\ccvu$ respectively
 denote the union and intersection of the covariates associated with each
 $t_i \in \trep$ i.e.,  $\ccvin= \bigcap \ccv_i$ and $\ccvu= \bigcup \ccv_i$. Furthermore,
 for each treatment $\trep \subseteq \tre$, define the view $\relt(ID,\ccvu,\ccvin,\trep,o)$ on $\rel$.
 Notice that in a special case where $\trep=\{t_i\}$, $\relei$ has the same schema as the view $\crele$ (cf. Section \ref{sec:cem}).

Lets define the view  $\prp$ on
 $\relt$ as follows:
}

Let $\trep \subseteq \tre$ be a subset of treatments
with $\ccvin= \bigcap_{T_i \in \trep} \ccv_{T_i} \not = \emptyset$. Without loss of generality assume $\trep=\{T_1 \ldots T_{k'}\}$.
 Consider the view $\prp$ over $\rele$ as shown in Figure \ref{fig:cf}(a). Essentially, $\prp$ describes CEM wrt. the disjunction
 of treatments in $\trep$ and the shared covariates between them.  \ignore{Intuitively, for any instance of $\rele$, and a
 subset of treatments $\trep \subseteq \tre$, computing the extension of $\prp$,
partitions the instance into
groups with similar covariates that are shared between the treatments in $\trep$ at least one $T_i \in \trep$. Next we show that, CEM wrt. each individual treatment in $\trep$ can be
obtained from $\prp$.} Figure \ref{fig:cf}(b) defines the view $\mcem$ over $\prp$ that describes a modified version of CEM for $T_i$,
based on the covariate factoring. \ignore{The following proposition shows CEM can be done by covariate factoring} \ignore{that CEM based on covariate factoring as
defined by $\mcem$ is equivalent to $\cem$ as described in Section \ref{sec:sub}.}

\vspace{-0.2cm}
\begin{proposition} \label{pro:prematching}
Given an  instance of $\rele$ and any subset of treatments $\trep\subseteq \tre$ with
$\bigcap_{T_i \in \trep} \ccv_{T_i} \not = \emptyset$, and for any $T_i \in \trep$, it holds that $\mcem(\rele)=\cem(\relei)$.
\end{proposition}
\vspace{-0.1cm}

\begin{proof}
  We sketch the proof for $S=\{T_1, T_2\}$. Covariate factoring on $S$, discard subclasses obtained from group-by
  on $\ccv_{T_1} \cap \ccv_{T_2}$ that have no overlap wrt. both $T_1$ and $T_2$. It is clear that group-by $\ccv_{T_1}$ is more fine-grained  than group-by on $\ccv_{T_1} \cap \ccv_{T_2}$, thus, subclasses with no overlap in the latter, form new subclasses in the former that still fail the overlap wrt. both of $T_1$ and $T_2$.
\end{proof}
\ignore{
Proposition \ref{pro:prematching} shows that CEM for multiple treatments can be performed by covariate factoring.}
Proposition \ref{pro:prematching} shows that CEM for multiple treatments can be performed by covariate factoring.
Next we develop a heuristic algorithm that takes advantage of this property to speed up CEM.
Before we proceed, we  state two observations that lead us to the algorithm.

First,  the total cost of performing  CEM independently for $k$ treatments
is a function of the size of the input database. However, the cost of performing the same task using covariate factoring
is a function of the size of the result of covariate factoring. But, the cost of covariates factoring
depends on the size of the input. Thus,  partitioning the treatments into a few set of groups, which results in pruning a huge
portion of the input database, reduces the overall cost of CEM.

\ignore{
First, by assuming a linear cost model for CEM,  the total cost of performing CEM independently for $k$ treatments
is $f(k*n)$, where $f$ is a linear function and $n$ is the number of rows in the input database. However, the cost
of computing CEM based on covariate factoring wrt. $m$ group of treatments,
is $f(k*n')$, where $n'$ is the size of the result of the covariates factoring, plus the cost of covariate factoring itself. Note that, we assumed that
this pre-processing equally prunes the input data for all groups. Since covariate factoring is sort of CEM we can refer to its cost
with the same function.  Thus, the cost of covariate factoring for $n$ groups is $f(m*n)$. Consequently, covariates factoring reduces the overall cost of
CEM if $f(k*n'+m*n)<f(k*n)$. Intuitively,  partitioning the treatments into a few set of groups, which results in pruning a huge portion of the input data, reduces the overall cost of CEM.
}

Second, we observe  that the correlation between the treatments  plays a crucial role in the  efficacy  of covariate factoring. The
 correlation between two treatments $T$ and $T'$ can be measured by the {\em phi coefficient} between them, denoted by $\phi$. Consider the following table

\begin{center}

\begin{tabular}{|c|c|c|c|}
\hline
   & $T=1$ & $T=0$ & Total  \\ \hline
  $T'=1$ & $n_{11}$ & $n_{10}$ & $n_{1\bullet}$ \\ \hline
  $T'=0$ & $n_{01}$ & $n_{00}$ & $n_{0\bullet}$ \\ \hline
  Total & $n_{\bullet1}$ & $n_{\bullet0}$ & $n$ \\
  \hline
\end{tabular}

\end{center}

\noindent  where $n_{11}$, $n_{10}$, $n_{01}$, $n_{00}$, are non-negative counts of number of observations that sum to $n$, the total number of observations.
 The phi coefficient between $T$ and $T'$ is given by
  $\phi(T,T') ={\frac  {n_{{11}}n_{{00}}-n_{{10}}n_{{01}}}{{\sqrt  {n_{{1\bullet }}n_{{0\bullet }}n_{{\bullet 0}}n_{{\bullet 1}}}}}}$. A phi coefficient of 0  indicates independence, while 1 and -1 indicates complete dependence between the variables (most observations falls off the diagonal).   Suppose $T_1$ and $T_2$ are highly correlated, i.e.,  $|\phi(T,T')|\simeq 1$ and share the covariates $X$. Further assume CEM wrt. $T$ and $X$, prunes 50\% of the input data.
Then, covariates factoring for $T_1$ and $T_2$ prunes almost 50\% of the input data. This is because, subclasses discarded by
CEM on $T$, are very likely to be discarded by CEM wrt. the disjunction of $T$ and $T'$ (because there are highly correlated). \ignore{
this variable is either 0 or 1, since $T'$ is highly correlated with $T$, it is very likely that most
of these subclasses are also discards
by CEM on $T'$. Therefore, covariates factoring should performed on correlated treatments.}

Algorithm \ref{al:cf1}, was developed based on these observations.  Section \ref{sec:opt} shows that covariates factoring based on this algorithm significantly reduces the over all cost of CEM wrt. multiple treatments in the \delay\ example.

\ignore{
their shared covariates
(not the entire set of covariate associated with them) similarly  retain $50\%$ of the input database.  The size of the result of
factorization is almost  $50\%$ of the input size the the correlation between $t_1$ and $t_2$ is almost 1;
on the other, if $t_1$ and $t_2$ are not correlated then pre-processing does not prune the data.

These observations lead us to the following algorithm based on the phi correlation matrix of the set of treatments.
}
\vspace{-0.2cm}\begin{algorithm}

\caption{Covariate Factoring}
\label{al:cf1}

1.\  Let $T_1 \ldots T_k$ be a set of treatments, and  $\ccv_{T_i}$ be a vector of covariates associated to $T_i$.   \newline
2.\ Construct a correlation matrix, $\mc{M}$, between the treatments such that
the $[i,j]$ entry in $\mc{M}$ contains $\phi(T_i,T_j)$.   \newline
3.\ Given a partition of the treatments into $n$ groups, $S_1 \ldots S_n$, such that $|S_k|\geq 2$ and
  $\bigcap_{T_i \in S_k} \ccv_{T_i} \not = \emptyset$, compute the normalized pairwise correlations in $s_k$
  as $\mc{C}_{S_k}= \frac{\sum_{(T_i,T_j) \in S_k} |\mc{M}[i,j]|}{|S_k|}$. \newline
5:\ Perform covariate factoring for groups obtained by the partition that maximises $\sum_{k=1\dots n} \mc{C}_{S_k}$.

\ignore{5:\ Find a partitioning that maximizes $\sum_{k=1\dots n} \mc{C}_{S_k}$ \ignore{, the sum of normalized pairwise correlations in all groups.} \newline
6:\ Perform the covariate factoring for each group of treatments obtained from this partitioning.}
\end{algorithm}

\subsubsection{(online)  Data-Cube Aggregation}
\label{sec:cube}

Given a set of observed attributes $\cv$ and an outcome, all attributes can be subjected to a causal
analysis. For example, all weather attributes can be dichotomized and form
a treatment. In addition, one may define other treatments, formed by conjunction of such dichotomized attributes. For instance, the conjunction of ``snow" and ``high-wind-speed" could become the ``snowstorm" treatment. Note that
causal effect is not subadditive \cite{janzing2013quantifying}. Thus, quantifying the causal effect of the conjunction of $T_1$ and $T_2$
requires an independent analysis on the treatment $T=T_1 \land T_2$ wrt. the covariates  $X_{T}= X_{T_1} \cup X_{T_2}$.

In principle, one might be interested
in exploring and quantifying the casual effect of $k=2^{|\cv|}$ treatments. In this setting, to estimate $\ate$ for all possible treatments, a matching method must be performed wrt. all possible subsets of $\cv$, each of which is associated to
one treatment. We argue that, in such cases, CEM for all treatments can be performed efficiently using the existing DBMS systems that support data-cube operations.

\ignore{
Recall from  \ref{fig:cem} that CEM is implemented with a group-by followed by a join (cf. Figure \ref{sec:cem}(c)).
More specifically, the rows (units) are grouped by the identified covariates. For each group three
aggregates namely $max(ID)$, $min(t)$
and $max(t)$ are computed. $max(ID)$ is used as a unique identifier for each subclass; the other
two are used to ensure that each group contains at least one treated and one control units.}

Recall that CEM for an individual treatment is a group-by operation (cf. Figure \ref{fig:cem}(b)). Thus, CEM for all treatments requires computing some aggregates on the data-cube ($\ccv$). \ignore{The idea is to modify the SQL
implementation of CEM (cf. Figure \ref{fig:cem}(b))  to perform the group-by statement using a relevant data-cube.} Now the established optimization techniques to
 compute data-cubes efficiently  can be adopted, e.g., for computing a group-by, we pick the
 smallest of the previously materialized groups from which it is possible to compute the group-by.
\ignore{For example, consider a four
covariates cube $(X_1,X_2,X_3,X_4)$.  Group by $X_1X_2$ can be obtained from $X_1X_2X_3X_4$,
$X_1X_2X_3$ and $X_1X_2X_4$. The idea is to pick the smallest between $X_1X_2X_3$, $X_1X_2X_4$. It is
apparent that both are smaller than $X_1X_2X_3X_4$.} In Section \ref{sec:opt}, we apply this idea to the \delay \ example and show that
it significantly reduces the overall cost of CEM for multiple treatments.

\vspace{-0.1cm}
\subsubsection{(offline) Databases Preparation for Causal Inference on Sub-populations}
\label{sec:dp}

So far, we considered causal inference as an online analysis which seeks to explore the effect of multiple treatments on an outcome, over a
population. In practice, however, one needs to explore and quantify the causal effect of
multiple treatments over various
sub-populations. For instance, \ignore{in the \delay \ example one needs to know:} what is the causal
 effect of low-visibility on departure delay in San Francisco International Airport (SFO)?  what is the effect of thunder at all airports in the state of Washington since 2010?
  such queries can be addressed by performing CEM on the entire data and selecting the relevant part of the obtained matched subset to the query.

Thus, the cost of performing CEM wrt. all possible treatments can be amortized over several causal queries. Therefore, we can prepare the database offline and pre-compute the matched subsets wrt. all
possible treatments to answer online causal queries efficiently. This could be impractical for high dimensional data since the number of possible treatments can be exponential in the number of attributes (cf. Section \ref{sec:cube}). Alternatively, we propose
Algorithm \ref{al:cf1}, which employs the covariate factoring and data-cube techniques to prepare the database so that CEM based on
any subset of the database can be obtained efficiently.

\ignore{
. Intuitively, it uses data-cubes to efficiently perform covariate factoring
for multiple treatments. For each group of treatments the result of covariate factoring is materialized and  proper data-cubes are constructed,
so that CEM for any individual treatment and over an arbitrary subset of a database can be efficiently obtained.  In Section \ref{sec:opt}, we apply this algorithm on the \delay \ example. We show that this algorithm
substantially reduces the cost of CEM and has a reasonable cost of preparation.}
\vspace{-0.2cm}
\begin{algorithm}
\caption{Database Preparation} \label{alg:dp}
1:\   Let $T_1 \ldots T_k$ be a set of treatments, and  $\ccv_{T_i}$ be a vector of covariates associated to $T_i$.   \newline
2:\ Apply  Algorithm \ref{al:cf1} to partition the treatments into $S_1 \ldots S_k$ with $\ccvu_{S_i}= \bigcup_{T_j \in S_i} \ccv_{T_j}$
and $\mc{X}'_{S_i}= \bigcap_{T_j \in S_i} \ccv_{T_j} $. \newline
3:\ Partially materialize $\mc{C}$, the cube  on $\ccv_1 \ldots \ccv_k$ to answer group-by queries for each $\mc{X}'_{S_i}$.\newline
4:\ For each group $S_i$, perform covariate factoring using $\mc{C}$ and materialize $P_{S_i}$. \newline
5:\ For each $P_{S_i}$, partially materialize $\mc{C}_i$, the cube on $\ccvu_{T_i}$, so that CEM for each $T \in g_i$ can be computed using $\mc{C}_i$.
\end{algorithm}

%% file: experimental.tex
\begin{figure*}
  \centering
    \includegraphics[width=1.1\linewidth]{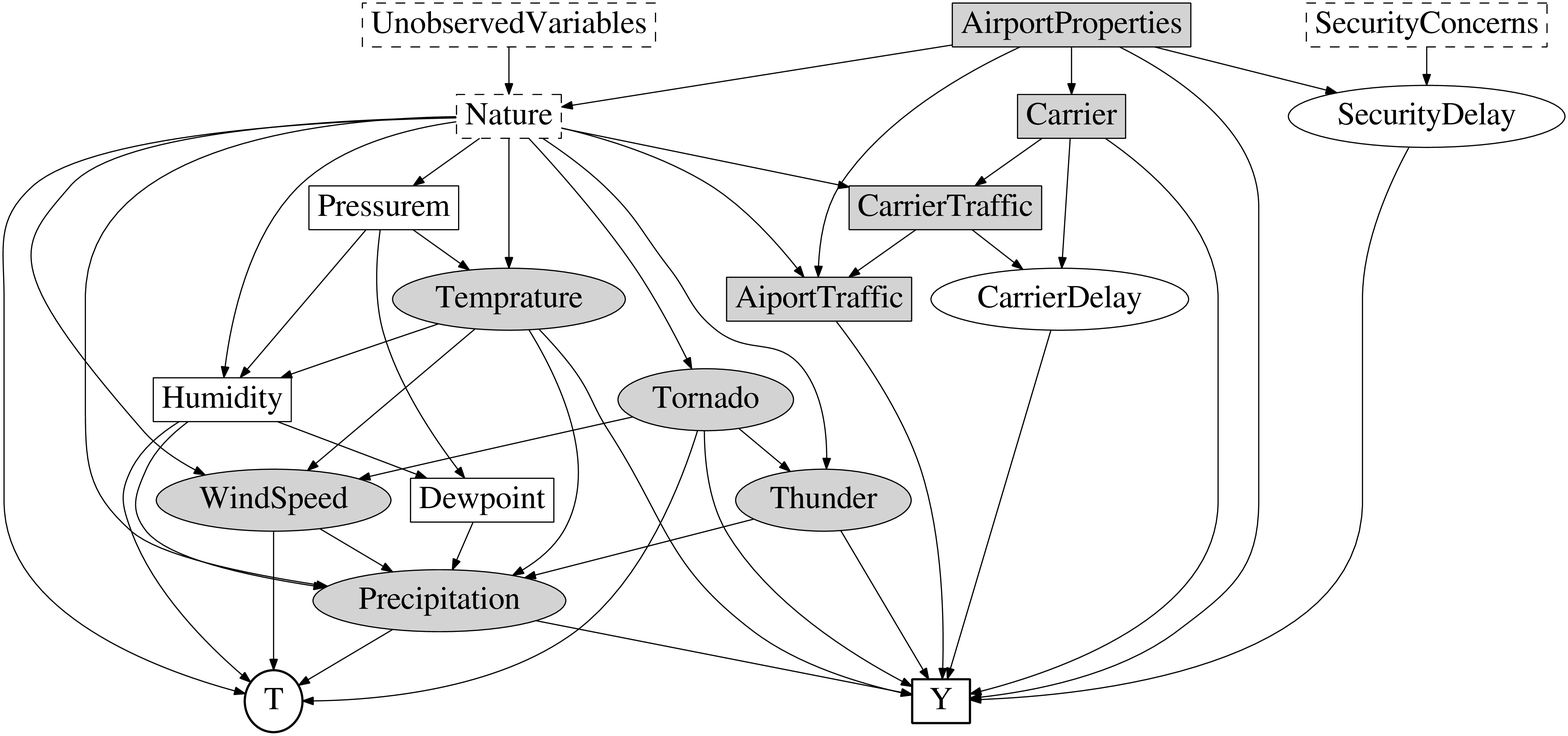}
\caption{ {\bf \small A Causal Directed Acyclic Graph (CDAG)  demonstrating
    the assumptions necessary to estimate $\ate$ of the treatment
    T(LowVisibility) to the outcome Y(DepDelay).}  A CDAG shows how conditioning on observable covariates (X)
breaks the confounding causal relationships (e.g., T  $\leftarrow$ Precipitation
$\rightarrow$ Y) and allows the estimation of $\ate$. In other words, we assume that the treatment assignment $T$
and potential outcomes $Y=Y(z)$ for z=0,1 are unconfounded
by conditioning on the filled nodes which is a minimum cardinality set of {\em observed} variables that $d$-sperate the two. Attributes in boxes are from the flight data; Those in ovals are from the weather data; Those in dashed boxed are unobserved.}
\label{fig:dag}
\end{figure*}

\section{Experimental Results}
\label{sec:exp}

We have implemented the basic techniques in
Sec.~\ref{sec:BasicTechniques} and the optimizations in
Sec.~\ref{sec:OptimizationTechniques} in a system called \GSQL.  This
section, presents experiments that evaluate the feasibility and
efficacy of \GSQL.  We addressed the following questions.  What is the
end-to-end performance of \GSQL \ for causal inference in
observational data?  Does \GSQL \ support advanced methods for causal
inference and produce results with the same quality as statistical
software? How does \GSQL\ scale up with increasingly large data sets?
And how effective are the optimization techniques in \GSQL?

\vspace{-.2cm}

\subsection{Setup}
\label{sec:setup}

% \subsubsection{Data}
% \label{sec:data}

{\bf Data} The {\em flight dataset} we used was collected by the US
Department of Transportation (U.S. DOT) \cite{flightdata}. It contains
records of more than 90\% of US domestic flights of major airlines
from 1988 to the present. Table \ref{tab:attlist}(shown in Section
\ref{sec:introduction}) lists dataset attributes that are relevant to
our experiments.  We restrict our analysis to about 105M data entry
collected between 2000 and 2015.

The {\em weather dataset} was gathered using the weather underground
API \cite{Weatherdata}.  Its attributes are also presented in Table
\ref{tab:attlist}. In addition, we pre-computed two other attributes
AiportTraffic and CarrierTraffic. The former is the total number of
flights occur in the origin airport of a flight one hour prior to the
flight departure time, the latter is the same quantity restricted to
the flights from the same carrier.  We managed to acquire and clean
35M weather observations between 2000 and 2015. These datasets are
integrated by a spatio-temporal join.

% \subsubsection{Causal questions and covariate selection}
% \label{sec:def}

{\bf Causal questions and covariate selection} We explore the causal
effect of the following binary treatments on flight departure delays
and cancellation: LowVisibility (1 if Visim$<1$; 0 if Visim$>5$); \
Snow (1 iff Precipm$>0.3$ and Tempm$<0$); \ WindSpeed (1 if
Wspdm$>40$; 0 if Wspdm$<20$); \ and Thunder.  In each case, data items
where the attribute in question was in between the two bounds were
discarded.  For instance, for the treatment LowVisibility, we want to
assess the following counterfactual: {\em ``What would the flight
  departure delay have been, if visibility were fine, i.e., Visim$>5$,
  when visibility is actually low, i.e., Visim$<1$"}; for this
analysis, items where Visim$\in[1,5]$ were discarded.

% For each analysis, all flights not in the treated or control groups are discarded (for Snow and Thunder, all units are either treated or control).

\ignore{
We used the RNCM to estimate the $\ate$ of each treatment on flight departure delay and cancellation via the matching methods. Each flight
form a unit in the RNCM.} To ensure the SUTVA (cf. Section \ref{subsec:causalitystatistics}),  we considered the difference between the
actual delay and the late aircraft delay (if one exists) as the
outcome of interest. Therefore, we assumed that the potential delay of each flight did not depend on the treatment assignment
to the other flights namely, there was no interference between the units. \ignore{ To avoid multiple version of a treatment, we only considered extreme weather conditions as treatment. For instance, visibility in the range less than $1 km$ is considered as
treatment, assuming that any value of visibility in this range would have
a similar effect on flight departure delay.}

To obtain quality answers for each treatment, we used graphical models
to identify a minimum number of covariate attributes to ensure
unconfoundedness, because minimizing the number of covariates has been
shown to increase the precision of the matching estimators
\cite{de2011covariate}.  We used the tool provided by
\url{http://dagitty.net} to construct the causal DAG (CDAG) and select
the minimum number of covariates that separate the treatment assignments from  the potential outcomes.  The tool finds a minimal subset of variables $X$ that forms a
{\em $d$-separation}~\cite{PearlBook2000} of the treatment $T$ from
the effect $Y$ (meaning: all paths from $T$ to $Y$ go through some
variable in $X$).  For example, Figure \ref{fig:dag} shows the CDAG
for the treatment LowVisibility and the effect DepDelay: the set $X$
consists of the shaded nodes d-sperate the treatment assignment and the potential outcomes.

\ignore{

 shown in
to represent the causal relationships between LowVisibility, other features, and the
departure delay. \ignore{\footnote{}} Informally, CDAG is a carrier of independent assumptions provided that two
random variables (represented as nodes of the graph) are assumed to be independent iff there is no edge
between them. Further conditional independence can be extracted from  a CDAG by applying
the {\em d-separation} criterion \cite{PearlBook2000}. Informally, the point
of $d$-separation is that, if all paths in the CDAG leading from one random variable (say, $T_i$) to another (say, $Y_i$) are
blocked by a certain subset of variables (say, $X_i$), then $T_i$ and $Y_i$ are conditionally
independent given $X_i$. \ignore{For instance for LowVisibility we obtain that following minimal set of observed covariates:
$X$=\{AiportTraffic, CarrierTraffic, Precipitation, Temperature, Thunder, Tornado, Carrier, Windspeed\}.}
For instance, for LowVisibility, the filled nodes in Figure \ref{fig:dag}, represent the set of covariates that unconfound LowVisibility and DepDelay. For more details on covariate selection using CDAGs, we refer the reader to \cite{de2011covariate} and references therein.

}

{\bf Systems} The experiments were performed locally on a 64-bit OS X machine with
Intel Corei7 processor (16 GB RAM, 2.8 GHz).  \GSQL\
was deployed to Postgres version 9.5.  We compared \GSQL\ with R packages MatchIt
and CEM, version 2.4 and 1.1 respectively,  available from \cite{CRAN}.

\subsection{Results}

\subsubsection{End-to-End Performance} \label{sec:endtoend}

We estimated the causal effect of different weather types on flight departure delay and cancellation at five major US airports, that are among the ten airports with the worse weather-related delays according to \cite{weather},  namely: San Francisco (SFO) , John F. Kennedy (JFK), Newark Liberty (EWR), George Bush (IAH), and LaGuardia Airport (LGA).  The
flight and weather data associated with these airports consists of about 10M and 2M rows, respectively.

To this end, we estimated $\ate$  for each treatment associated to a weather type by performing CEM wrt. its identified covariates.  Continuous covariates were coarsened into equal width buckets, and categorical covariates were
matched by their exact values.

Figure \ref{fig:eteresult}(a) reports the running time of performing
CEM for each treatment: recall (Fig.~\ref{fig:cem}) that this involves
a group-by followed by an elimination of all groups that have no
treated, or no control unit.
\ignore{
 In this figure, the number of
treated and matched units for each treatment is represented as
labels. As depicted, CEM produced matched samples with reasonable
sizes.} Next, we evaluated the quality of the CEM, in Figure
\ref{fig:eteresult}(b), using a standard metric in the literature: the
{\em absolute weighted mean difference (AWMD)} of each continuous
covariate between the treated and control group:

{
\begin{align}
& \E_b[|\E[x_i | T=1, B(X)=b]-\E[x_i | T=0, B(X)=b]|] \label{eq:awmd}
\end{align}
}\noindent for each covariate attribute $x_i \in X$.  This is
a standard measure of imbalance, since in a perfectly balanced group
the difference is 0 (see Eq. \ref{eq:bs}). Note that in the case of CEM, we assume subclasses form the balancing scores.  We compared this
difference for the original value of the covarites  before and after CEM.  As shown, CEM substantially reduced
the covariates imbalance in the raw data: this graph shows the perils
of attempting to do causal inference naively, without performing CEM
or matching.  We also observed that CEM results in quite reasonable
matched sample size wrt. treatments ( more than 75\% of the treated
units are matched with the average rate of one treated to five control
units).

 \ignore{ we performed CEM in
  which the continuous covariates are simply coarsened by
  rounding. Figure \ref{fig:eteresult}(b,d) since data is pretty dense
  this simple coarsening strategy produced well-balanced matches with
  reasonable sizes. Figure \ref{fig:eteresult}(b,d) represents the
  estimation of $\ate$ for each treatments before and after CEM.}

\ignore{
. That is we measure the mean
difference of a covariate between the treated and control group within
each subclass and take the average of this quantity across all
subclasses.} \ignore{It is trivial that the matched sample is perfectly
  balanced on the categorical covariates such as Carrier.}

Next, Figure \ref{fig:eteresult}(c) shows the causal effect of weather
type on departure delay: it shows the estimated $\ate$ for each target
airport, normalized by the number of treated units to the total number
of units in each aiport.  Following common practice in causal
inference we estimated $\ate$ using Eq. \ref{eq:oatt} (thus assuming
that CEM produces perfectly balanced groups). The normalization lets
us compare the effects of different weather types in an individual
airport. In fact, $\ate$ reflects the potential causal effect of a
treatment to an outcome. For instance, IAH barely experiences snow,
but snow has a devastating effect on flight delay at this airport. The
results in Figure \ref{fig:eteresult}(c) are in line with those in
\cite{weather}, which reported the following major weather-related
causes of flight delay at the airports under the study:
snowstorms,thunderstorm and wind at EWR; thunderstorm and fog at IAH;
snowstorms and visibility at JFK; snowstorms at LGA; fog and low
clouds at SFO;

Figure \ref{fig:eteresult}(d) a similar causal effect of weather type
on a different outcome, namely Cancellation. Here, we leveraged the
fact that matching and subclassification do not depend on a particular
outcome variable, thus, a matched sample wrt. a treatment can be used
to estimate its causal effect of several outcomes.

\ignore{The credibility of the findings can be assessed by comparing the obtained $\ate$ for each treatment with the average difference
of weather delay and NAS delay between the treated and control group, we call
 this quantity {\em actual delay bound (ADB)}, as recorded in the flight dataset.
 ADB provides an upper bound for the actual causal effect of a weather type to departure delay.
As depicted in Figure \ref{fig:eteresult}(a), the estimated $\ate$ are always less than the ADB. In a similar way we measured
{\em actual cancellation bound (ACB)}, and compare it to the estimated effect of weather types on cancellation.}

\begin{figure*}
\hspace*{-0.7cm}\begin{subfigure}{0.62\linewidth}
\centering
\includegraphics[height=5.5cm,width=\linewidth]{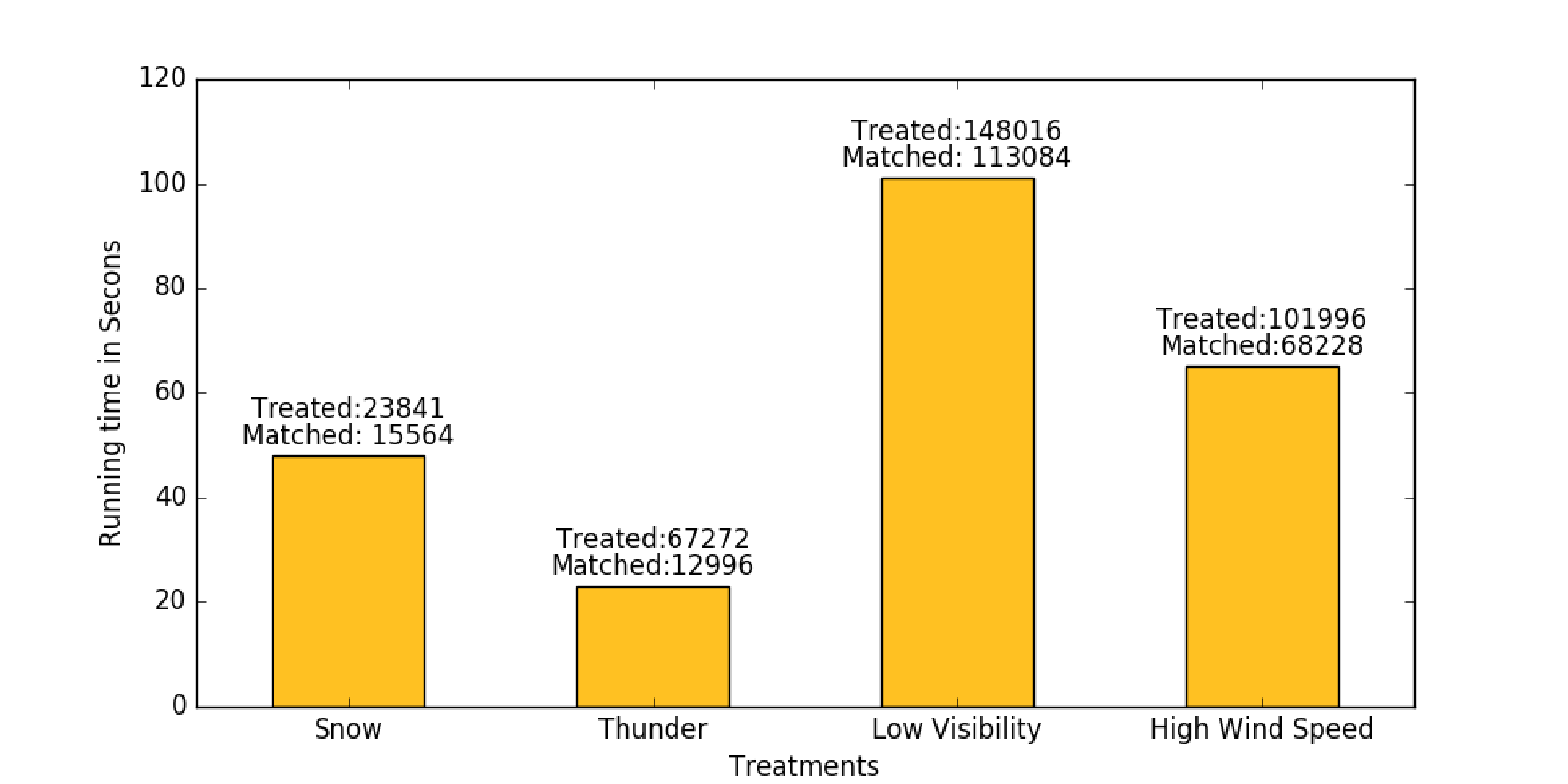}
\caption{{Running Time of Matching.}}
\label{sfig:testa}
\end{subfigure}\hfill
\hspace*{-0.7cm}\begin{subfigure}{0.62\linewidth}
\centering
\includegraphics[height=5.5cm,width=\linewidth]{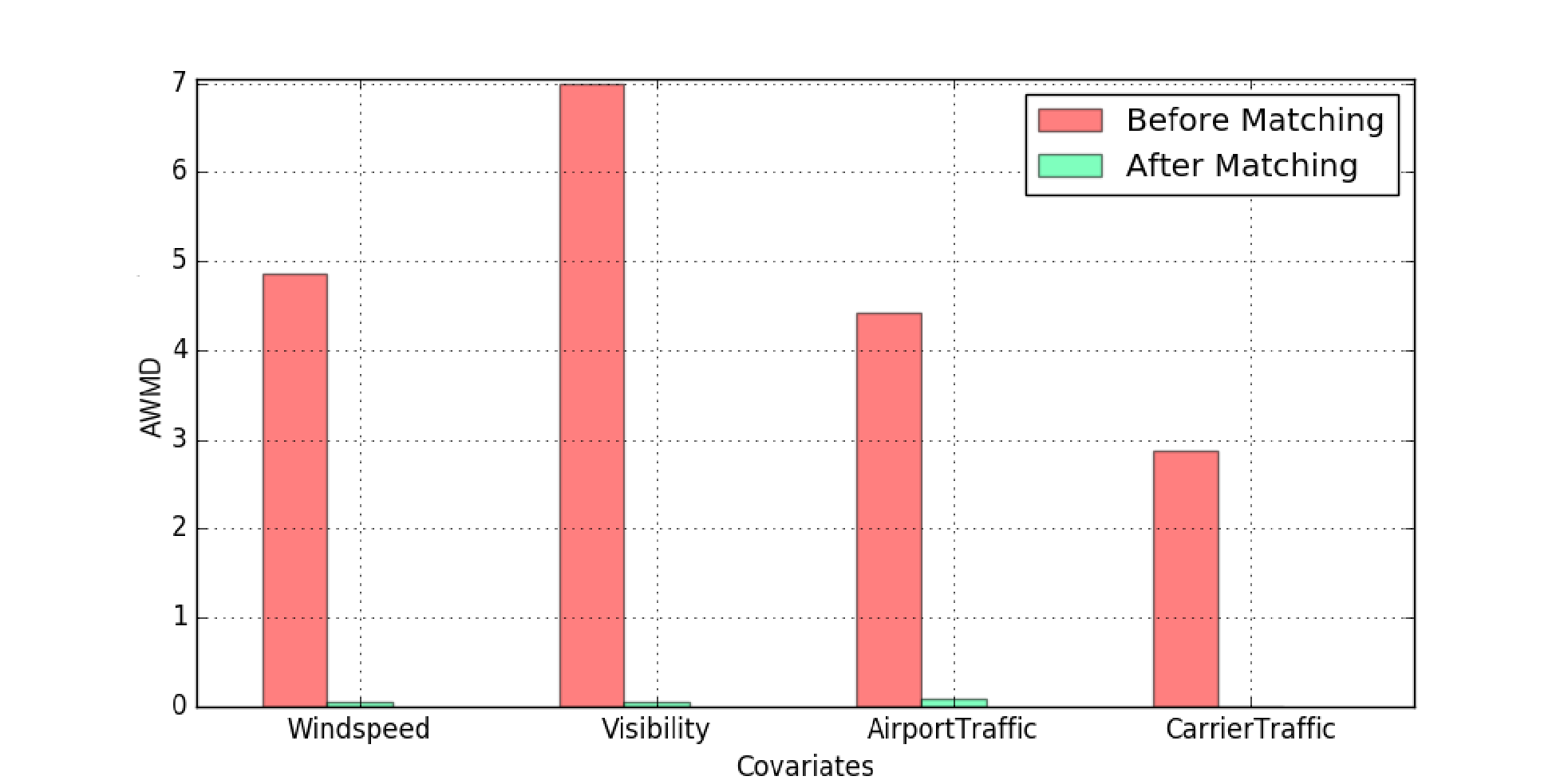}
\caption{{Measuring Imbalance Reduction.}}
\label{sfig:testb}
\end{subfigure}\hfill

\hspace*{-0.7cm} \begin{subfigure}{0.62\linewidth}
\centering
\includegraphics[height=5.5cm,width=\linewidth]{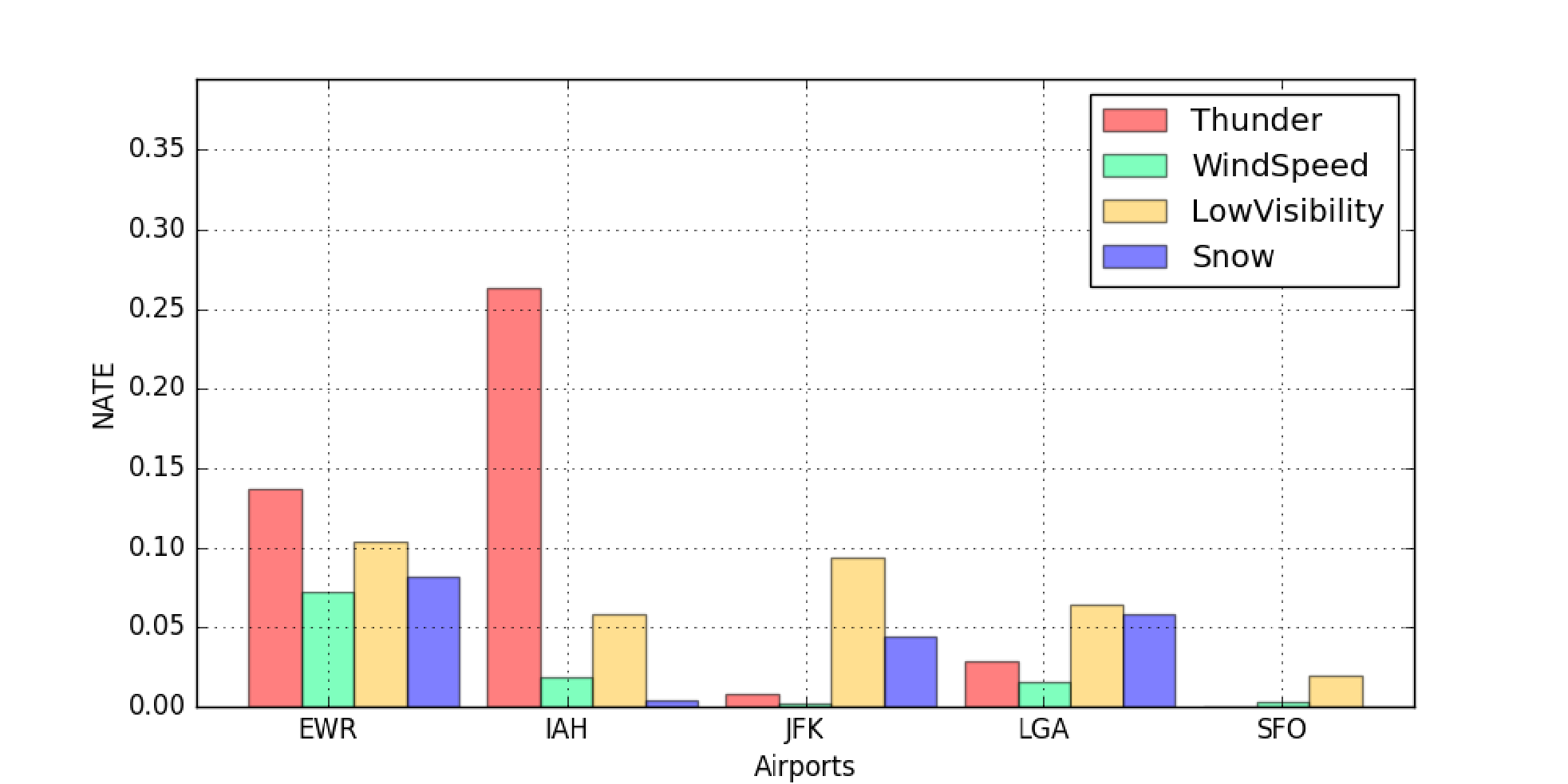}
\caption{{Causal Effect of Weather Types on  Delay.}}
\label{sfig:testc}
\end{subfigure}\hfill
\hspace*{-0.7cm}\begin{subfigure}{0.62\linewidth}
\centering
\includegraphics[height=5.5cm,width=\linewidth]{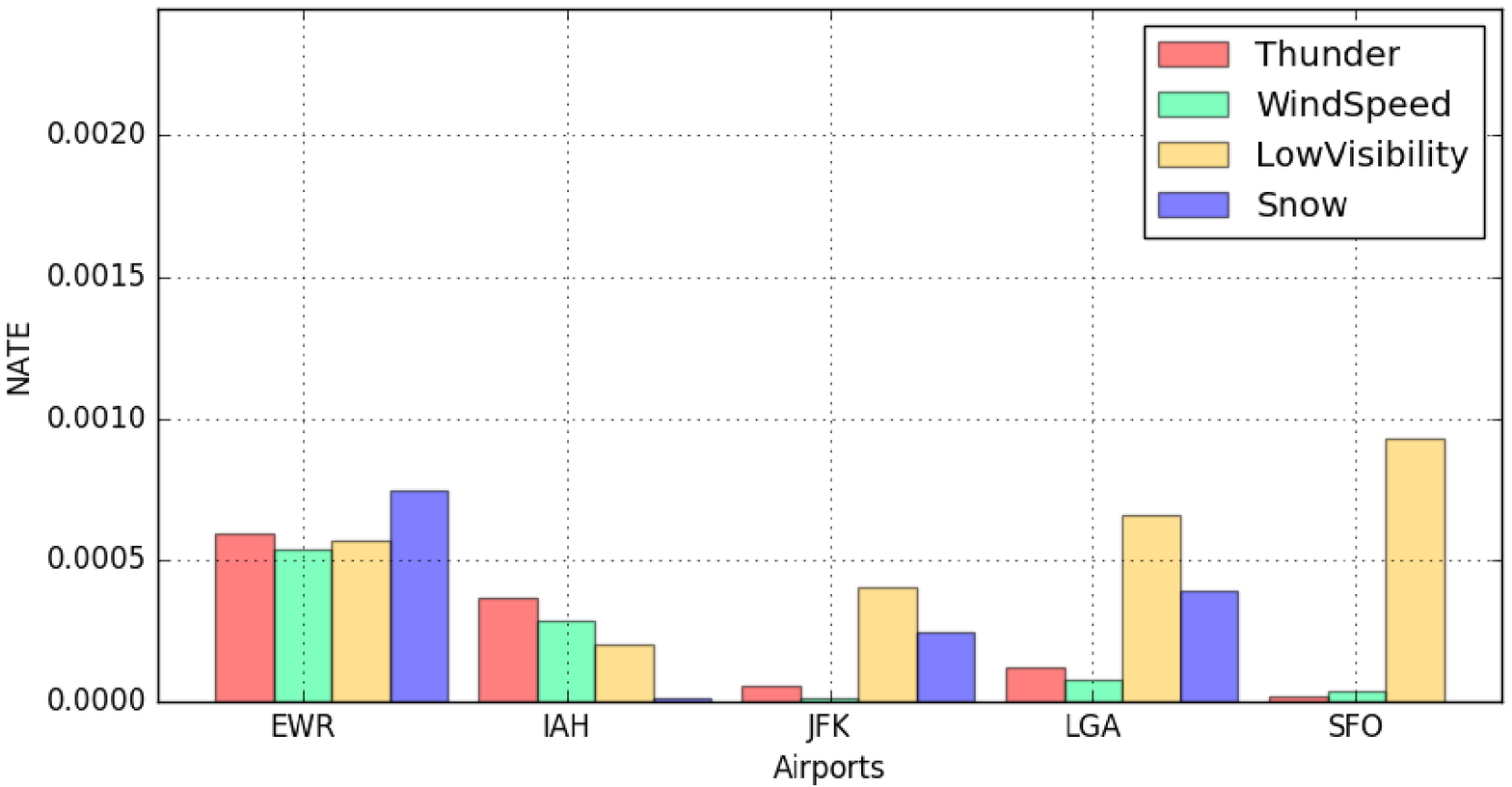}
\caption{{Causal Effect of Weather Types  on  Cancellation.}}
\label{sfig:testd}
\end{subfigure}\hfill

\caption{ \bf{Analysis of the causal effect of weather on flight delay and cancellation.}}
\label{fig:eteresult}
\end{figure*}

\vspace{-.3cm}
\subsubsection{Quality Comparisons with R}

% We compared the quality of the matching methods provided in \GSQL  \ with their
% counterpart methods in R. \ignore{More specifically, we compare exact matching,
% CEM, 1-1 NNMNR , 1-1 NNMWR and Subclassification
% based on propensity score with the similar methods
% provided in the MatchIt and CEM packages in R.}

The MatchIt and CEM packages in R are popular tools for conducting
causal inference today.  In this section we compare the quality of the
results returned by \GSQL\ with those returned by the R packages.  We
considered all types of matchings described in
Sec.~\ref{sec:BasicTechniques}: NN matching (both NNMNR and NNMWR),
and subclassification (by propensity score, CEM and exact matching(EM)).  Sine the R
packages do not scale to large datasets, we conducted these
experiments over a random sample of data used in Section
\ref{sec:endtoend}, which consists of around 210k rows. We evaluated
different matching methods for the treatment of Snow.

Table \ref{tbl:qc} compares the size of the matched sample and the
AWMD (Eq.\ref{eq:awmd}) obtained by \GSQL \ and R, for different
matching methods.  These results shows that \GSQL\ produced results
whose quality was at least as good as R.  Note that for NNMNR and
NNMWR, the caliper 0.001 is applied.   For subclassification all units with
propensity score less than 0.1 and more than 0.9 are discarded (this
is a common practice in causal inference). We observe that all
matching methods produce well-balance matched samples with reasonable
sizes. However, among these methods CEM works better both in terms of
the imbalance reduction and size of the matched sample.

The slight differences between
the results of NNM arose from a slight
difference between the propensity score distribution we obtained using logistic regression provided by MADlib inside Postgres. \ignore{ In addition, observe that subclassification in \GSQL \ did better than
it counterpart in R. This is because, }
%\cite{Hellerstein2012}

\ignore{
As represented in Table \ref{tbl:qc}, applying this caliper still provide us with a pretty reasonable sample size. Tables

This would increase the performance of
our NNM matching methods. We note that more sophisticated technics for performing NNM in SQL exists e.g., \cite{d}. However,
propensity score matching was not the main concern of our paper.}

\begin{table*}[t]
  \centering  \scriptsize
  \begin{tabular}{|c|c|c|c|c|c|c|c|} \hline
    & \multicolumn{3}{c|}{}   & \multicolumn{4}{c|}{\bf{Absolute Weighted Mean Difference (AWMD)}} \\ \hline
   \bf{Method} &  & \bf{Control} & \bf{Treated}  & \bf{Visibility} & \bf{WindSpeed}& \bf{AirportTraffic} & \bf{CarrierTraffic} \\ \hline
    & { \em Raw Data }
 & 214457 & 464 & 12.7529& 3.6363
&4.5394&2.9465 \\ \hline
 \raisebox{-.5\normalbaselineskip}[1pt][1pt]{\rotatebox[origin=c]{0}{\bf{NNMWR}}}&
R & 311 & 323
   &  0.0724
 & 0.8789
 &0.6935  &0.2724
 \\
 &\GSQL			& 296		&		312	&	 0.0692&	  0.5756&	0.6955	&0.8044
\\ \hline
 \raisebox{-.5\normalbaselineskip}[1pt][1pt]{\rotatebox[origin=c]{0}{\bf{NNMNR}}}&   R & 318
 & 318
 & 0.1006
 &1.0308
 & 0.3396&0.1352
\\
& \GSQL &	 291		&		291&0.0769&	0.7216& 	0.3195	& 0.5910
\\ \hline
 \raisebox{-.5\normalbaselineskip}[1pt][1pt]{\rotatebox[origin=c]{0}{\bf{Subclass.}}} &
 R &   1275 &255 & 2.5054
  & 1.4631
 & 1.6013&1.0022
\\
& \GSQL &1002 &
255& 0.0684&  1.0631&  0.1872&  0.0905
\\ \hline
 \raisebox{-.5\normalbaselineskip}[1pt][1pt]{\rotatebox[origin=c]{0}{\bf{EM}}} &      R & 8
 & 7
 & 0
 & 0
 & 0&0
 \\
 & \GSQL & 8
 & 7
 & 0
 & 0
 & 0&0
 \\ \hline
  \raisebox{-.5\normalbaselineskip}[1pt][1pt]{\rotatebox[origin=c]{0}{\bf{CEM}}} &     R &2284&
340&
0.2875 &  0.0542&   0.1135& 0.0905 \\
& \GSQL &2284&
340&
0.28755 &  0.0542&   0.1135& 0.0905 \\ \hline

    \end{tabular} \vspace{0.5cm}
    \caption{\bf{Quality Comparison between \GSQL \ and R.}}
  \label{tbl:qc}
\end{table*}

\vspace{-.25cm}
\subsubsection{Scalability Testing}
\label{sec:sct}
We compared the scalability of different matching methods in \GSQL \ and R. The experiment was carried out over
a random samples of data used in Section \ref{sec:endtoend} and for the treatment of Snow. Figure \ref{fig:perfresults}(a) compares NNM methods based on propensity score. We observe that NNM does not scale to large data. Note that, for the reasons mentioned in Section \ref{sec:nnm}, optimizing this method was not the aim of this paper. Figure \ref{fig:perfresults}(b)  compares CEM, EM and subclassification in \GSQL \ and R. Note that for CEM and NNMWR in \GSQL\, we respectively implemented the group-by and window function statement. As depicted,  \GSQL \ significantly outperforms R and scale to large data.

\begin{figure*}
\hspace*{-0.7cm}\begin{subfigure}{0.59\linewidth}
\centering
\includegraphics[height=5cm,width=1.02\linewidth]{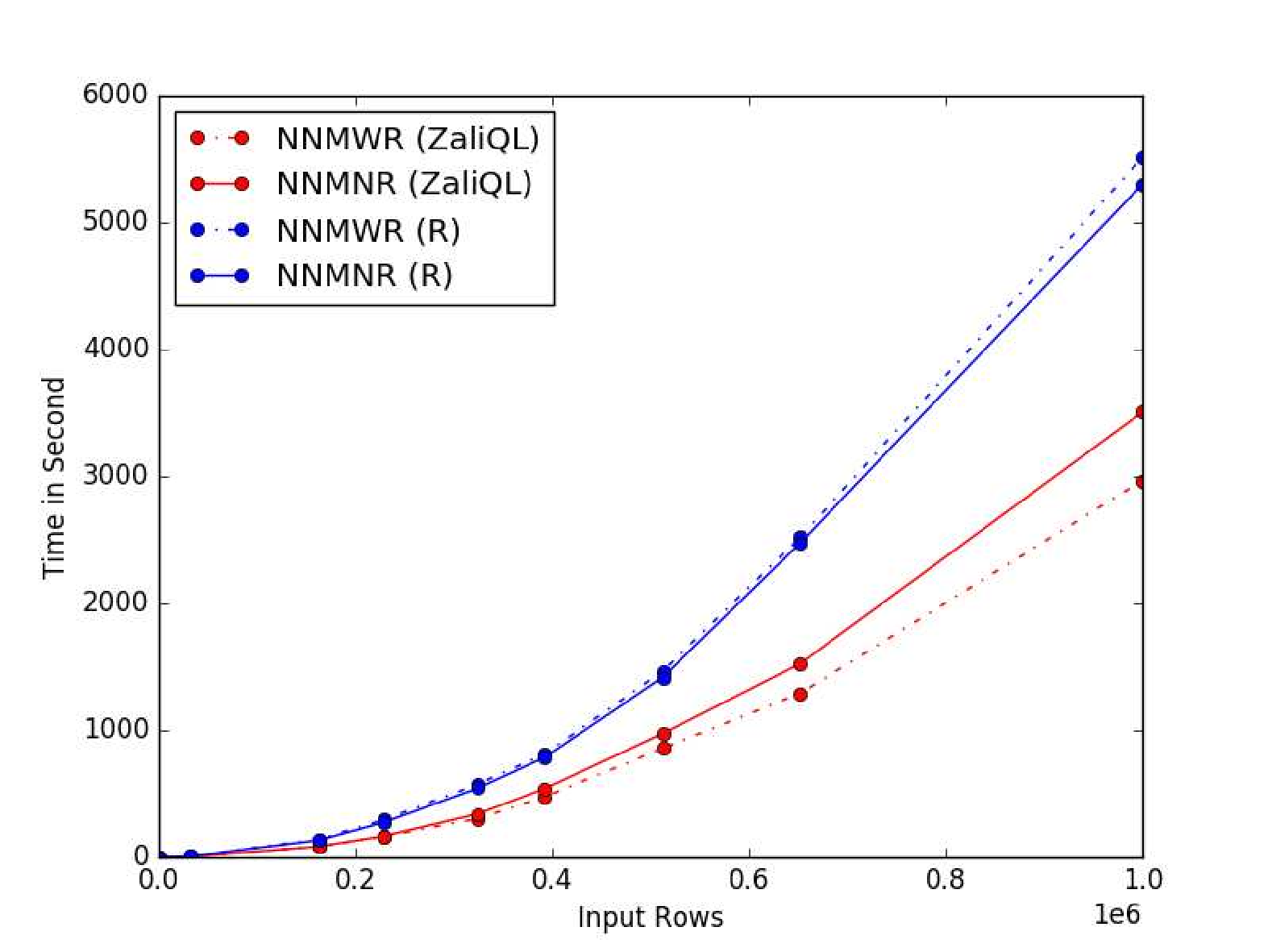}
\caption{{Scalability Comparison (NNM).}}
\label{sfig:testaa}
\end{subfigure}\hfill
 \hspace*{-0.45cm}\begin{subfigure}{0.59\linewidth}
\centering
\vspace*{0.41cm} \includegraphics[height=5cm,width=1.02\linewidth]{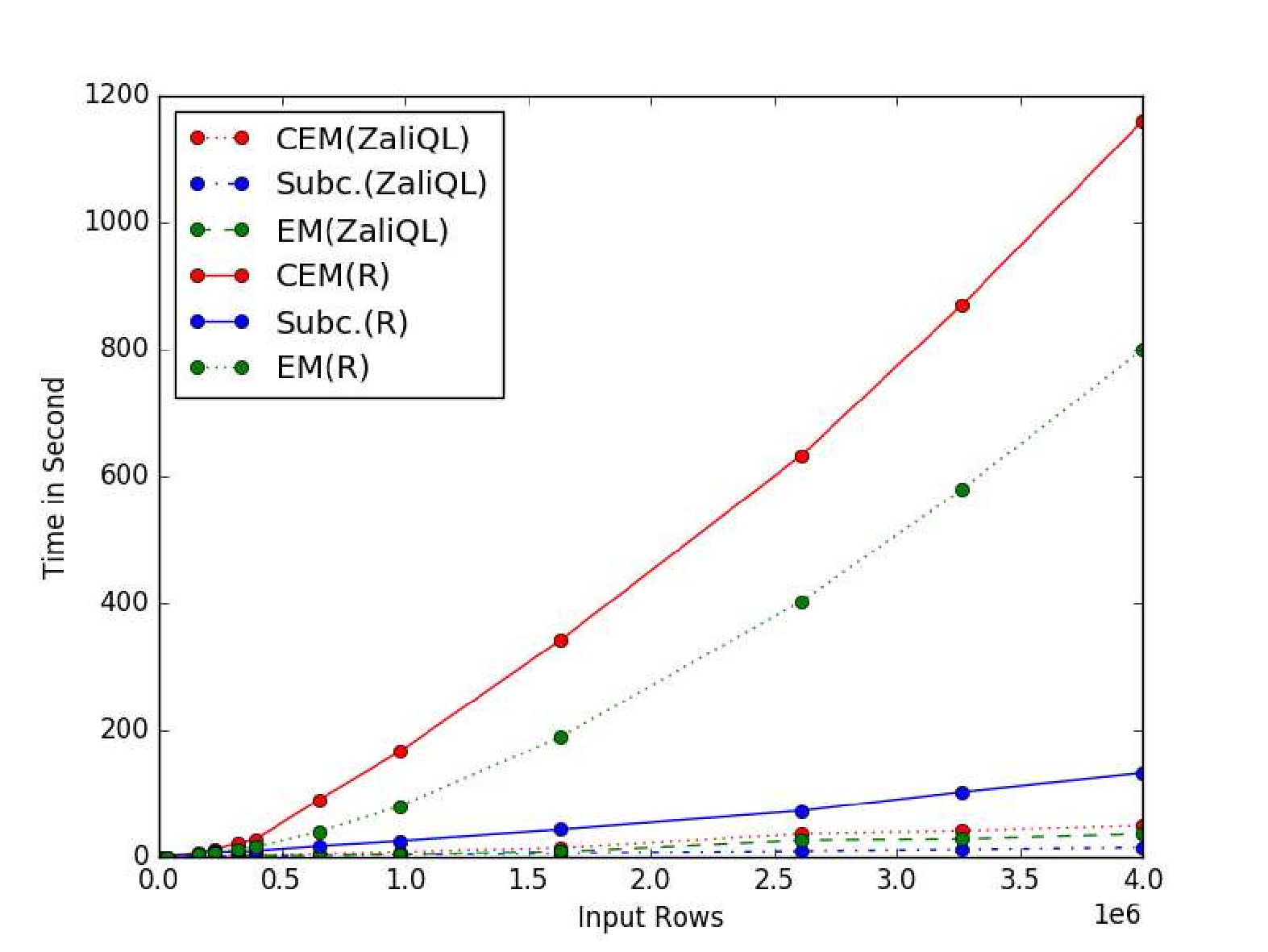}
\caption{{Scalability Comparison (CEM, EM and Subclas.).}}
\label{sfig:testbb}
\end{subfigure}\hfill

\hspace*{-0.55cm}\begin{subfigure}{0.59\linewidth}
\centering
\includegraphics[height=5.5cm,width=1.02\linewidth]{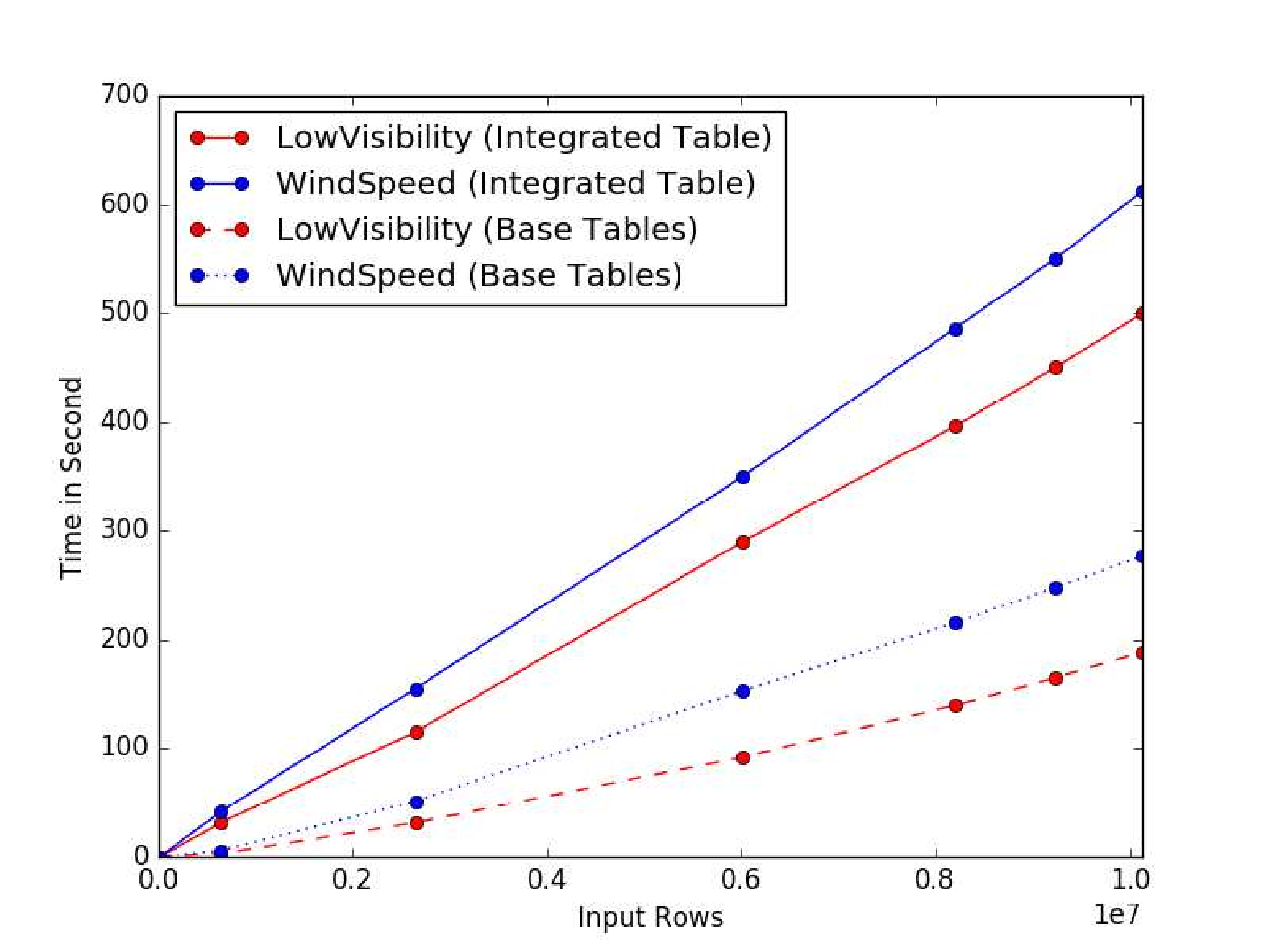}
\caption{Efficacy of Pushing Matching.}
\label{sfig:testbc}
\end{subfigure}\hfill
\hspace*{-0.5cm}\begin{subfigure}{0.59\linewidth}
\centering \vspace{0.4cm}
\includegraphics[height=5.5cm,width=1.02\linewidth]{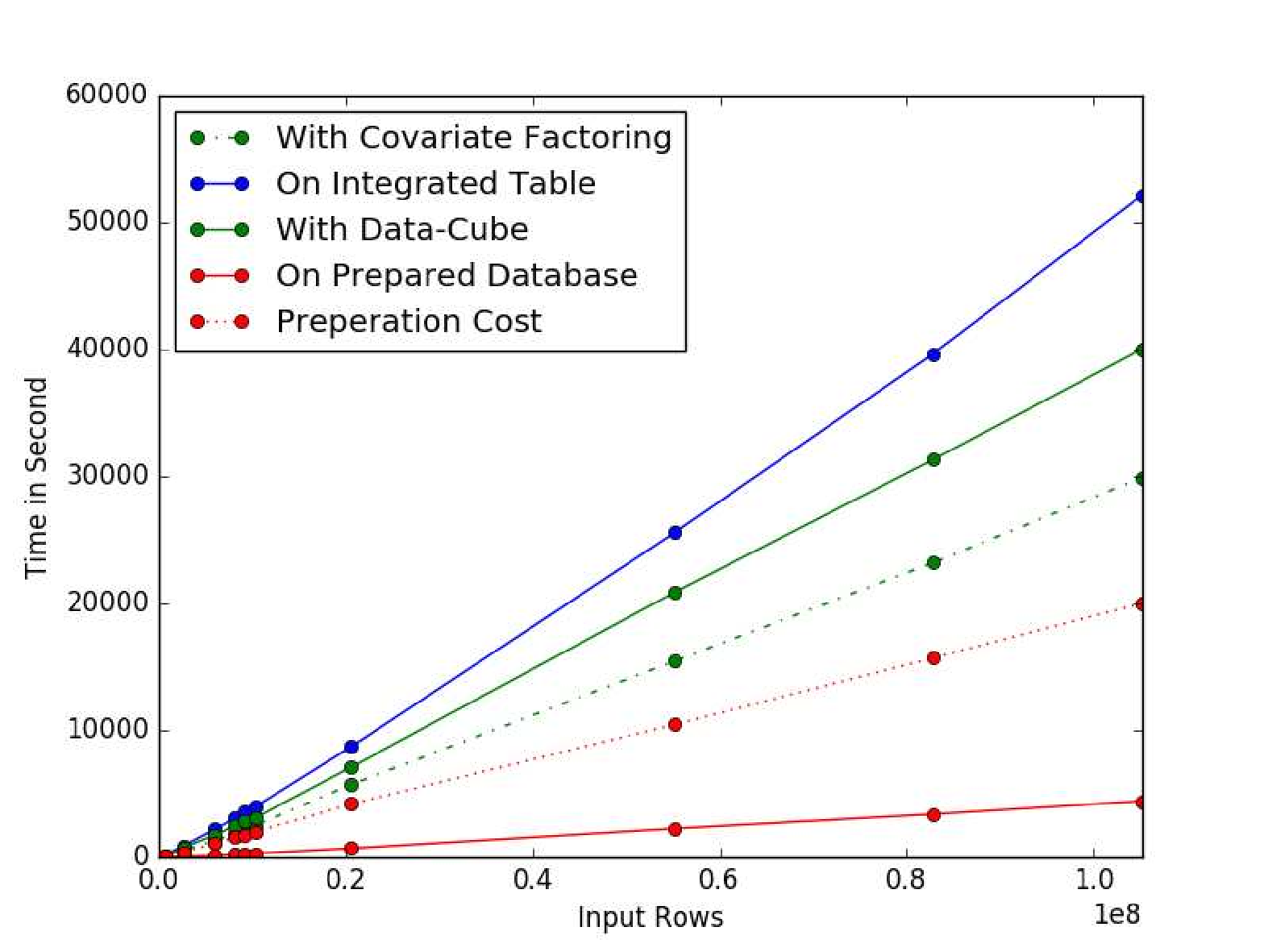}
\caption{Efficacy of optimizations for multiple treatments.}
\label{sfig:testbd}
\end{subfigure}\hfill

\caption{\bf{Scalability and optimizations Evaluation.}}
\label{fig:perfresults}
\end{figure*}

\vspace{-.25cm}
\subsubsection{Efficacy of the Optimization Techniques}
\label{sec:opt}
We tested the effectiveness of the proposed optimization techniques (cf. Section \ref{sec:OptimizationTechniques}). Figure \ref{fig:perfresults}(c) compares the running time of performing CEM on
the integrated weather and flight tables with CEM on base tables (cf. Section \ref{sec:baserel})
for two treatment LowVisibility and WindSpeed.  The analysis was carried out on data used in Section \ref{sec:endtoend}. Note that the cost of integrating two tables is ignored.

For covariates factoring and data-cube optimizations, we compared the
total cost of performing matching on the treatments defined in Section
\ref{sec:setup}, plus the treatment obtained by conjunction of Snow and
WindSpeed, which we refer to as Snowstorm. This analysis was carried
on the entire integrated dataset. By applying Algorithm \ref{al:cf1},
the treatments are partitioned into two groups, namely $g_0$=\{Snow,
WindSpeed, Snowstorm\} and $g_2$=\{LowVisibility, Thunder\}.  Figure
\ref{fig:perfresults}(d) compares the total running time of matching
for each treatment after covariate factoring with the naive matching
procedure. Figure \ref{fig:perfresults}(d)  also shows the total cost of matching with and without using data-cubes.
 In addition, it represents the total cost of matching on the prepared database, according to Algorithm \ref{alg:dp}, and the cost of database preparation. As depicted,  matching on the prepared database
reduce the over cost of matching by an order of magnitude.  \ignore{Observe that the total cost of performing matching on the prepared database plus the cost of preparation is less than matching using only data-cube and covariates factoring. \ignore{Therefore, Algorithm \ref{alg:dp} can also be applied to the setting that }}

\ignore{
{\em Matching on bases tables:} Figure \ref{fig:perfresults}(b) compares the running time of performing CEM on
the pre-integrated weather and flight tables with CEM on base tables (cf. Section \ref{sec:baserel}). The comparison
is made on different sample sizes. Note that the cost
of joining two tables is ignored.

{\em Factorization:} we evaluate the factorization technique to  optimize the overall cost of performing matching
for $T_1$ \ldots $T_4$. To this end first we apply the Algorithm
and obtained the following treatment partitions.  [...] Figure \ref{fig:perfresults}(c)

{\em Data-cubes:} We used the data-cube technique to optimize the overall cost of performing matching
for $T_1$, $T_2$, $T_3$ and any treatments
that can be defined by conjunction of these treatments e.g., we can define $T_{12}$ as 1 if $T_1$ and $T_2$
are 1 and 0 otherwise. Therefore, we have seven treatments in total. Figure  \ref{fig:perfresults}(d) compares the over cost of performing matching for these seven treatments independently and computing them based on the ton-down approach using
data-cubes.
}

%% file: related.tex
\vspace{-.42cm}
\section{Related work and conclusion}
%\dan{7 pages}
\label{sec:rel}

\ignore{
The study of causality in statistics can be traced back to Neyman~\cite{neyman1923},
Fisher~\cite{Fisher1935design}, Rubin~\cite{Rubin1974,Rubin2005} and
Holland~\cite{Holland1986}, whose theories  led to the {\em potential outcome framework} for causal
inference. More recently, Pearl~\cite{PearlBook2000}, Spirtes et al.~\cite{Spirtes:book01} and
others have developed causal graphical models, which provide a graphical way to describe
causal relationships and support reasoning over multiple causal relationships. However \ignore{
In both cases, causal inference  critically depends on the precise methodology and on key assumptions about the data,
without which causal relationships cannot be claimed.}  See \cite{richardson2013single} for a connection between these two approaches.}
The simple nature of the RNCM, and its adherence to a few statistical assumptions, makes it more appealing for the researchers. Therefore,
 it has become the  prominent approach in social sciences, biostatistics, political science, economics
 and other disciplines. Many toolkits have been developed for performing casual inference  {\em $\acute{a}$ la} this framework that
 depends on  statistical software such as SAS, SPSS, or R project. However, these toolkits do not scale to large datasets. This work introduce  \GSQL,  a SQL-based framework for drawing causal
  inference that circumvents the scalability issue with the existing
  tools. ZaliQL supports state-of-the-art methods for causal
  inference and runs at scale within a database engine.

Causality has been studied extensively in databases \cite{DEBulletin2010,MeliouGMS2011,RoyS14,DBLP:conf/icdt/SalimiB15,SalimiTaPP16,DBLP:conf/flairs/SalimiB16,DBLP:conf/uai/SalimiB15}.  We note that this line of work is different than the present paper in the sense that, it aims to identify causes for an observed output of a data transformation. \ignore{For example, in query-answer causality/explanation, the objective is to identify parts of a database that causally explain a result of a query. } While these works share some aspects of the notion of causality as studied in this paper, the problems that they address are fundamentally different. 